\documentclass{article}
\pdfoutput=1

\usepackage[nonatbib,final]{neurips_2024}

\usepackage{graphicx}
\usepackage{placeins}
\usepackage{tikz}
\usetikzlibrary{arrows,decorations.pathreplacing, positioning}
\usepackage{pgfplots, pgfplotstable}
\usetikzlibrary{shapes.geometric, arrows, fit, positioning, backgrounds}
\usetikzlibrary{arrows.meta,patterns}
\pgfdeclarelayer{background}
\pgfdeclarelayer{foreground}
\pgfsetlayers{background,main,foreground}
\pgfplotsset{compat=1.18}
\usepgfplotslibrary{fillbetween}
\pgfplotsset{
    layers/my layer set/.define layer set={
        background,
        pre main,
        main,
        foreground
    }{        %
    },
    set layers=my layer set,
}
\tikzset{
  pics/car/.style args={}{
  code={
    \begin{scope}[
      every node/.style={
        draw,
        font=\footnotesize,
        node distance=3mm and 4mm
      }
    ] 
            \draw[very thick, rounded corners=0.5ex,fill=c15,thick]  (0,.15) -- ++(2,0) -- ++(0,0.25) -- ++(-0.5,0.05) -- ++(-0.25,+0.25) -- ++(-0.75,0.0) -- ++(-0.25,-0.2) -- ++(-0.25, 0.) -- (0,.15);
            \draw[draw=black,fill=gray!50,thick] (0.4,.165) circle (.165);
            \draw[draw=black,fill=gray!50,thick] (1.5,.165) circle (.165);
            \draw[draw=black,fill=black,thick] (0.4,.165) circle (.025);
            \draw[draw=black,fill=black,thick] (1.5,.165) circle (.025);
        \end{scope}
    }
  }
}
\tikzset{
dot/.style = {circle, fill, minimum size=#1,
              inner sep=0pt, outer sep=0pt},
dot/.default = 6pt %
}

\definecolor{c1}{HTML}{f9f871}
\definecolor{c2}{HTML}{ffc75f}
\definecolor{c3}{HTML}{ff9671}
\definecolor{c4}{HTML}{ff6f91}
\definecolor{c5}{HTML}{d65db1}
\definecolor{c6}{HTML}{845ec2}

\definecolor{c11}{HTML}{008f7a}
\definecolor{c12}{HTML}{008e9b}
\definecolor{c13}{HTML}{0089ba}
\definecolor{c14}{HTML}{0081cf}
\definecolor{c15}{HTML}{2c73d2}
\definecolor{c16}{HTML}{845ec2}

\definecolor{c21}{HTML}{c4fcef}

\usepackage{xspace}

\usepackage[breaklinks=true]{hyperref}

\usepackage{tablefootnote}

\usepackage{acro}

\usepackage{amsmath}
\usepackage{amssymb}
\usepackage{amsthm}
\usepackage{mathtools}
\usepackage{stmaryrd}
\usepackage{mleftright}

\usepackage{cleveref}
\crefname{formula}{formula}{formulas}
\creflabelformat{formula}{#2(#1)#3}
\crefformat{footnote}{#2\footnotemark[#1]#3}

\usepackage{tipa}

\newtheorem{definition}{Definition}
\newtheorem{lemma}[definition]{Lemma}
\newtheorem{theorem}[definition]{Theorem}

\usepackage{relsize}
\usepackage{accents} 

\usepackage{multirow}

\usepackage{sepfootnotes}

\usepackage{floatrow}
\newfloatcommand{capbtabbox}{table}[][\FBwidth]

\usepackage{MnSymbol}

\usepackage{dsfont}

\usepackage{siunitx}

\usepackage{algorithm}
\usepackage[noEnd=true,indLines=true,italicComments=false,spaceRequire=false]{algpseudocodex}

\usepackage{tabularx}
\usepackage{wrapfig}

\usepackage{cite}

\usepackage[inline]{enumitem}

\usepackage[createShortEnv,conf={}]{proof-at-the-end}

\usepackage{microtype}

\newcommand{\chFourFancyName}{Mosaic}
\newcommand{\chThreeFancyName}{VerSAILLE}

\usepackage{todonotes}

\title{
Provably Safe Neural Network Controllers\\via Differential Dynamic Logic
}
\author{%
  Samuel Teuber\textsuperscript{1} \quad
  Stefan Mitsch \textsuperscript{2} \quad
  Andr\'e Platzer \textsuperscript{1,3}\\ %
  \textsuperscript{1} Karlsruhe Institute of Technology \quad
  \textsuperscript{2} DePaul University \quad
  \textsuperscript{3} Carnegie Mellon University\\ %
  \texttt{teuber@kit.edu}\quad
  \texttt{smitsch@depaul.edu}\quad
  \texttt{platzer@kit.edu}\\
}

\newcommand{\formulaset}[2][]{\ensuremath{\text{FOL}_{#2}^{#1}}}

\newcommand{\truesym}{\ensuremath{\top}}
\newcommand{\reals}{\ensuremath{\mathbb{R}}}
\newcommand{\linreals}{\ensuremath{\mathrm{L}\mathbb{R}}}

\newcommand{\nnProgramFml}{\ensuremath{\alpha_g}}%
  
\newcommand{\noetherianReals}{\ensuremath{\mathrm{N}\mathbb{R}}}%

\newcommand{\atomsof}[1]{\ensuremath{\text{Atom}\mleft(#1\mright)}}

\newcommand{\grammar}[2]{\ensuremath{#1\ \Coloneqq{}\ #2}}

\newcommand{\interpretationsubstitute}[3]{\ensuremath{#1_{#2}^{#3}}}

\newcommand{\varsofformula}[1]{\ensuremath{\text{V}\mleft(#1\mright)}}

\newcommand{\dL}[0]{{\text{\upshape\textsf{d{\kern-0.05em}L}}\xspace}}
\newcommand{\dLstatespace}[0]{\ensuremath{\mathcal{S}}}
\newcommand{\dLinterpretationapp}[1]{\ensuremath{\left\llbracket #1 \right\rrbracket}}
\newcommand{\boundedvars}[1]{\ensuremath{BV\mleft(#1\mright)}}

\newcommand{\freevars}[1]{\ensuremath{FV\mleft(#1\mright)}}

\newcommand{\defaultcontroller}[0]{\ensuremath{\left(\alpha_{\text{ctl}};\alpha_{\text{plant}}\right)^*}}
\newcommand{\hybridsystemscontract}[0]{\ensuremath{%
\phi \implies \left[\defaultcontroller{}\right]\psi%
}}

\newcommand{\satatoms}[1]{\text{sat-atoms}\mleft(#1\mright)}
\newcommand{\projection}[2]{#1\!\downharpoonright_{#2}}

\newcommand{\controllerDescriptionFml}{\ensuremath{\kappa}}
\newcommand{\nnControllerDescriptionFml}{\ensuremath{\kappa_g}}

\newcommand{\controllerFormula}{\ensuremath{\zeta_c}}
\newcommand{\stateFormula}{\ensuremath{\zeta_s}}
\newcommand{\realArithmeticSafetyGuarantee}[0]{\ensuremath{\stateFormula \land \controllerDescriptionFml \implies \controllerFormula}}
\newcommand{\safetyGuaranteeProgramNoLoop}[0]{\alpha_{\text{refl}}\mleft(\controllerDescriptionFml\mright); \alpha_{\text{plant}}}
\newcommand{\safetyGuaranteeProgram}[0]{\left( \safetyGuaranteeProgramNoLoop \right)^*}
\newcommand{\safetyGuarantee}[0]{\phi \implies \left[ \safetyGuaranteeProgram \right] \psi}
\newcommand{\completeStateSpaceFormula}[0]{\ensuremath{\left(\phi \land \langle \left(\alpha_{\text{ctl}};\alpha_{\text{plant}}\right)^* \rangle \bigwedge_{i=1}^m z_i = z_i^+\right)
\implies
\left(\stateFormula\right)_{z_1 \ldots z_m}^{z_1^+ \ldots z_m^+}}}

\newcommand{\accRelPos}[0]{\ensuremath{p_{\text{rel}}}}
\newcommand{\accRelVel}[0]{\ensuremath{v_{\text{rel}}}}
\newcommand{\accRelAcc}[0]{\ensuremath{a_{\text{rel}}}}
\newcommand{\accVelConst}[0]{\ensuremath{v_{\text{const}}}}
\renewcommand{\implies}{\rightarrow}
\renewcommand{\iff}{\leftrightarrow}
\DeclareMathOperator*{\reluop}{ReLU}
\newcommand{\relu}{\ensuremath{\reluop}}

\DeclareAcronym{OLNNV}{
    short = open-loop NNV,
    long = open-loop neural network verification,
    first-style=short
}
\DeclareAcronym{CLNNV}{
    short = closed-loop NNV,
    long = closed-loop neural network verification,
    first-style=short
}
\DeclareAcronym{NN}{
    short = NN,
    long = neural network
}
\DeclareAcronym{FNN}{
    short = FNN,
    long = feed forward neural network
}
\DeclareAcronym{NNCS}{
    short = NNCS,
    long = neural network based control system
}
\DeclareAcronym{DNF}{
    short = DNF,
    long = disjunctive normal form
}
\DeclareAcronym{ACASX}{
    short = ACAS X,
    long = Airborne Collision Avoidance System X,
    cite={olson2015airborne}
}
\DeclareAcronym{ACASXu}{
    short = ACAS Xu,
    long = Airborne Collision Avoidance System X unmanned
}
\newcommand{\keymaeraxtext}{Ke{\kern-0.75ptY}maera X}
\DeclareAcronym{keymaerax}{
    short = \keymaeraxtext{},
    long = \keymaeraxtext{},
    first-style = long,
    tag = noindexplease
}
\DeclareAcronym{modelplex}{
    short = ModelPlex,
    long = ModelPlex,
    first-style = long,
    tag = noindexplease
}
\DeclareAcronym{DNNV}{
    short = DNNV,
    long = DNNV,
    cite={Shriver2021},
    first-style = long,
    tag = noindexplease
}
\DeclareAcronym{OVERT}{
    short = OVERT,
    long = OVERT,
    first-style = long,
    tag = noindexplease
}
\DeclareAcronym{dL}{
    short={\text{\upshape\textsf{d{\kern-0.05em}L}}\xspace},
    long=differential dynamic logic
}
\DeclareAcronym{NMAC}{
    short=NMAC,
    long=Near Mid-Air Collision
}
\DeclareAcronym{CPS}{
    short=CPS,
    long=Cyber-Physical System
}
\DeclareAcronym{SNNT}{
    short=\textsc{N$^3$V},
    long=\textbf{N}on-linear \textbf{N}eural \textbf{N}etwork \textbf{V}erifier,
    first-style=short
}
\DeclareAcronym{SMT}{
    short=SMT,
    long=Satisfiability Modulo Theories
}
\DeclareAcronym{RL}{
    short=RL,
    long=reinforcement learning
}
\DeclareAcronym{MILP}{
    short=MILP,
    long=Mixed Integer Linear Programming
}
\DeclareAcronym{ACC}{
    short=ACC,
    long=Adaptive Cruise Control
}
\DeclareAcronym{TCAS}{
    short=TCAS,
    long=Traffic Alert and Collision Avoidance System
}
\DeclareAcronym{FAA}{
    short=FAA,
    long=Federal Aviation Administration
}

\begin{document}

\maketitle

\begin{abstract}
\looseness=-1
While \acfp{NN} have a large potential as autonomous controllers for Cyber-Physical Systems, verifying the safety of \emph{\acfp{NNCS}} poses significant challenges for the practical use of \acsp{NN}---especially when safety is needed for \emph{unbounded time horizons}.
One reason for this is the intractability of analyzing \acp{NN}, ODEs and hybrid systems.
To this end, we introduce \chThreeFancyName{} %
(\textbf{Ver}ifiably \textbf{S}afe \textbf{AI} via \textbf{L}ogically \textbf{L}inked \textbf{E}nvelopes)%
:
The first general approach that allows reusing control theory literature for \ac{NNCS} verification.
By joining forces, we can exploit the efficiency of \ac{NN} verification tools while retaining the rigor of \ac{dL}. 
Based on a provably safe control envelope in \ac{dL}, we derive a specification for the \ac{NN} which is proven with \ac{NN} verification tools.
We show that a proof of the \ac{NN}'s adherence to the specification is then \emph{mirrored} by a \ac{dL} proof on the infinite-time safety of the \ac{NNCS}.

The \ac{NN} verification properties resulting from hybrid systems
typically contain \emph{nonlinear arithmetic} over formulas with \emph{arbitrary logical structure} while efficient \acs{NN} verification tools merely support linear constraints.
To overcome this divide, we present \chFourFancyName{}: An efficient, \emph{sound and complete} verification approach for polynomial real arithmetic properties on piece-wise linear \acp{NN}.
\chFourFancyName{} partitions complex NN verification queries into simple queries and
lifts off-the-shelf linear constraint tools to the nonlinear setting in a completeness-preserving manner by combining approximation with exact reasoning for counterexample regions.
In our evaluation we demonstrate the versatility of \chThreeFancyName{} and \chFourFancyName{}:
We prove infinite-time safety on the classical Vertical Airborne Collision Avoidance \ac{NNCS} verification benchmark for some
scenarios while (exhaustively) enumerating counterexample regions in unsafe scenarios.
We also show that our approach significantly outperforms the State-of-the-Art tools in \ac{CLNNV}.

\end{abstract}
\section{Introduction}
\label{sec:introduction}
\looseness=-1
For controllers of \acfp{CPS}, the use of \acfp{NN} is both a blessing and a curse.
On the one hand, using \acp{NN} allows the development of goal-oriented controllers that optimize soft requirements such as passenger comfort, frequency of collision warnings or energy efficiency.
On the other hand, guaranteeing that \emph{all} control decisions chosen by an \ac{NN} are safe is very difficult due to the complex feedback loop between the subsymbolic reasoning of an \ac{NN} and the intricate dynamics often encountered in physical systems.
How can this curse be alleviated?
Neural Network Verification (NNV) techniques all have tried one of three strategies:
Open-loop NNV entirely omits the analysis of the physical system and only analyzes input-output properties of the NN~\cite{zhang2018efficient,xu2020automatic,xu2020fast,wang2021beta,Henriksen2020,Henriksen2021,katz2017reluplex,katz2019marabou,bunel2020branch,bak2020improved}.
Open-loop analyses alone cannot justify the safety of an \ac{NNCS}, because they ignore its physical, feedback-loop dynamics.
Closed-loop NNV performs a time-bounded analysis of the feedback loop between the \ac{NN} and its physical environment~\cite{Forets2019JuliaReachReachability,Schilling2021,tran2019star,tran2020neural,9093970,IvanovACM21,Ivanov2021,Huang2019,Fan2020,Dutta2019,Sidrane2021,Akintunde2022}.
Unfortunately, a safety guarantee that comes with a time-bound (measured in seconds rather than minutes or hours) is often insufficient when it comes to deploying safety-critical \acp{NNCS} in the real world.
For example, the safety of an adaptive cruise control system must be independent of the trip length.
Finally, another line of work explored techniques for learning and then verifying \emph{approximations of} barrier certificates for infinite-time guarantees~\cite{DBLP:conf/ijcai/BacciG021,DBLP:conf/nips/LechnerZCH21,DBLP:conf/ciss/ChenFMPP21,DBLP:conf/nips/ChangRG19,DBLP:conf/corl/DawsonQGF21,DBLP:journals/corr/abs-2301-11683}.
For \emph{continuous-time}, verification has not been scaled beyond simple linear control functions~\cite[Appendix]{DBLP:conf/nips/ChangRG19} as it requires \ac{OLNNV} w.r.t. nonlinear specifications, which is a notoriously neglected topic~\cite{DBLP:conf/corl/DawsonQGF21,DBLP:journals/trob/DawsonGF23}.

\looseness=-1
As an alternative to the three outlined approaches, we propose to verify \acp{NNCS} based on the rigorous mathematical foundations of \acf{dL}.
\ac{dL} is a program logic allowing the proof of infinite-time safety for abstract, nondeterministic control strategies (often called control envelopes).
Due to its expressiveness and its powerful proof calculus, \ac{dL} even allows the derivation of such guarantees for continuous-time systems or systems whose differential equations have no closed-form solution.
By grounding our verification approach in \ac{dL}, we can reuse safety results from the control theory literature for NN verification -- especially for cases where characterizations of safe behavior and controllable/invariant regions are known (e.g. airborne collision avoidance~\cite{Julian2019a}).
How this knowledge can be reused is a non-trivial question:
While \ac{dL} is an excellent basis for reasoning about symbolic control strategies, the numerical/subsymbolic reasoning of NNs at their scale is far beyond the intended purpose of \ac{dL}'s proof calculus.
Conversely, open/closed-loop NNV tools and barrier certificates lack the infinite-time and exact reasoning available within \ac{dL}.
This work demonstrates how open-loop NNV can be combined with \ac{dL} reasoning to combine their strengths while canceling out their weaknesses.
Consequently, by relying on results from the control theory literature, we prove infinite-time safety guarantees for \acp{NNCS} that are not provable through either technique alone.

\begin{figure}[t]%
    \centering
    \resizebox{!}{3cm}{
    \begin{tikzpicture}[node distance=1cm]
        \tikzstyle{io} = [minimum height=0.6cm, text centered]
    	\tikzstyle{process} = [rectangle, minimum width=1.5cm, minimum height=0.6cm, text centered, draw=black]
    	\tikzstyle{arrow} = [thick,->,>=stealth]

        \node [rectangle,fill=black!5!white,minimum width=14.5cm,minimum height=1.7cm,anchor=north west] at (0cm,0.2cm) {};
        \node [rectangle,fill=c1!50!white,minimum width=14.5cm,minimum height=1.1cm,anchor=north west] at (0cm,-1.5cm) {};
        \node [rectangle,fill=c1,minimum width=14.5cm,minimum height=1.3cm,anchor=north west] at (0cm,-2.6cm) {};

        \node (envelope) at (0.1cm,0cm) [io,anchor=north west] {Control Envelope};

        \node (keymaerax) [process,right=of envelope,xshift=-0.6cm] {\ac{dL} model (proven in Ke{\kern-0.75ptY}maera X)};

        \node (esafe) [io, right=of keymaerax,color=c11,xshift=5.35cm] {\large Envelope safe};

        \node (modelplex) [process,below=of keymaerax.north west,anchor=north west,yshift=0.2cm] {ModelPlex};

        \node (versaille) [process,below=of modelplex.north west,anchor=north west] {Open-Loop NNV query};

        \node (nn) [io,right=of versaille.east, anchor=west] {NN $g$};

        \node (mosaic1) [process,right=of versaille,yshift=-1.2cm,xshift=-3.75cm,align=left] {Approximation \&\\DPLL(T) generalization}; %

        \node (mosaic2) [process,right=of mosaic1,xshift=-0.5cm] {Open-Loop NNV};

        \node (mosaic3) [process,right=of mosaic2,align=left,xshift=-0.5cm] {SMT reasoning\\(for completeness)};

        \node (nncssafe) [io,right=of versaille.east,below=of esafe.west,anchor=west,yshift=-0.8cm,color=c11, draw=c11,double] {\large NNCS safe};

        \node (nnprop) [io,right=of mosaic3.east,below=of nncssafe.west,anchor=west,color=c11,yshift=-0.2cm] {\large NN property holds};

        \draw [arrow] (envelope) -- (keymaerax);
        \draw [arrow] (envelope.east) -- (modelplex.west);
        \draw [arrow] (keymaerax) -- (esafe);
        \draw [arrow] ([xshift=1.9cm]keymaerax.south west) -- ([xshift=1.9cm]versaille.north west);
        \draw [arrow] ([xshift=1cm]modelplex.south west) -- ([xshift=1cm]versaille.north west);
        \draw [arrow] ([xshift=0.5cm]versaille.south west) |- (mosaic1.west);

        \draw [arrow] (mosaic1) -- (mosaic2);

        \draw [arrow] (mosaic2) -- (mosaic3);

        \draw [arrow] (mosaic3) -- (nnprop);

        \draw [arrow] (nn) -- (versaille);

        \draw [-{Implies},double,line width=0.4mm,color=c11] ([xshift=0.5cm]esafe.south west) -- ([xshift=0.5cm]nncssafe.north west);
        \draw [-{Implies},double,line width=0.4mm,color=c11] ([xshift=0.5cm]nnprop.north west) -- ([xshift=0.5cm]nncssafe.south west);

        \draw [-] (0cm,-1.5cm) -- ++(14.5cm,0cm);

        \node [io,anchor=south west,align=left] at (0cm,-1.5cm) {Preliminaries};

        \node [io,anchor=north west,align=left] at (0cm,-1.5cm) {\Cref{sec:safe_controller_implementations}:};

        \node [io,anchor=south west,align=left] at (0cm,-2.6cm) {\large VerSAILLE};

        \draw [-] (0cm,-2.6cm) -- ++(14.5cm,0cm);
        \node [io,anchor=north west,align=left] at (0cm,-2.6cm) {\Cref{sec:nnnv}:};

        \node [io,anchor=south west,align=left] at (0cm,-3.6cm) {\large Mosaic};

        \draw [arrow,dashed,color=black!80] (nn) -- node[anchor=south,color=black!80] {provide semantics for} (nncssafe);
        \draw [dashed,color=black!80] ([yshift=-0.2cm]keymaerax.east) -| ([xshift=1.1cm]nn.east);
    \end{tikzpicture}
}
    \caption{\chThreeFancyName{} reflects a proof of a control envelope in an \ac{NN} to verify infinite-time safety of an \ac{NNCS} from mere \ac{OLNNV} properties.
    \chFourFancyName{} \emph{completely} lifts off-the-shelf \ac{OLNNV} tools to polynomial arithmetic by combining approximation with judicious SMT reasoning.}
    \label{fig:overview}
\end{figure}
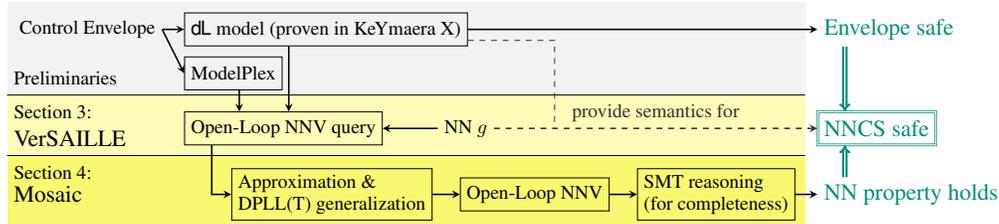
\paragraph{Overview.}
\looseness=-1
This paper alleviates the curse of \ac{NNCS} safety. %
As shown in \Cref{fig:overview}, our work integrates the deductive approach of \ac{dL} with techniques for \ac{OLNNV}.
To apply our approach, we assume that an abstract, nondeterministic control envelope has already been verified in \ac{dL} (via \ac{keymaerax}~\cite{Fulton2015}, synthesized via CESAR~\cite{DBLP:conf/tacas/KabraLMP24} or from the literature).
Based on the \ac{dL} safety result,
\chThreeFancyName{} (\textbf{Ver}ifiably \textbf{S}afe \textbf{AI} via \textbf{L}ogically \textbf{L}inked \textbf{E}nvelopes; \Cref{sec:safe_controller_implementations}) derives a verification query for \ac{OLNNV} by instrumenting \ac{modelplex}~\cite{Mitsch2016}.
By reflecting the \ac{NN} through a \emph{mirror program} in \ac{dL}, we can then reason about an \ac{NNCS} in- and outside the \ac{dL} calculus simultaneously.
The verification of an \ac{OLNNV} query generated by \chThreeFancyName{} yields a \ac{dL} proof that the \ac{NNCS} refines a safe control envelope
---implying that the infinite-time safety guarantee carries over to the \ac{NNCS}.

Due to the inherent nonlinearities of hybrid systems, the generated \ac{OLNNV} queries often contain polynomial arithmetic and the formulas have arbitrary logical structure.
Hence for such queries, we also introduce \chFourFancyName{}---an efficient, \emph{sound and complete} framework for \ac{OLNNV} tools.
The approach \emph{lifts} complete off-the-shelf \ac{OLNNV} tools for linear constraints to \emph{polynomial} constraints with \emph{arbitrary} logical structure.
To this end, we combine approximation 
with a generalization of DPLL(T) that makes the logical decomposition \emph{efficiently} applicable to
\ac{NN} verification 
(whereas ``classical'' DPLL(T) would become prohibitively inefficient).
At the same time, \chFourFancyName{} retains completeness by generalizing counterexamples into locally affine regions (\Cref{sec:nnnv}).
In summary, \chThreeFancyName{} provides rigorous semantics and a formal proof of infinite-time safety, while \chFourFancyName{} makes our approach practically applicable to real-world systems (see also \Cref{sec:evaluation}).

\paragraph{Contribution.}
Our contribution has three parts.
While our implementation (\ac{SNNT}) supports \acp{NN} most commonly analyzed by \ac{OLNNV} (\relu{} \acp{NN}), our theoretical contribution (\chThreeFancyName{}) reaches far beyond this and lays the foundations for analyzing a wide range of \ac{NNCS} architectures:
\begin{itemize}
    \item We present \chThreeFancyName{},
    the formal foundation that, for the first time, enables a \emph{sound proof of infinite-time safety for a concrete \ac{NNCS}} by reusing safety proofs from control-theory literature (in the form of \ac{dL} models).
    \chThreeFancyName{} supports a large class of feed-forward \acp{NN} (any \ac{NN} with piece-wise Noetherian activation functions; see \Cref{sec:background}).
    \item
    We introduce \chFourFancyName{}, a framework for the efficient, \emph{sound and complete} verification of properties in \emph{polynomial} real arithmetic on piece-wise linear \acp{NN}.
    Unlike other \ac{NN} verifiers,
    \chFourFancyName{} furthermore supports constraints with \emph{arbitrary} propositional structure.
    \chFourFancyName{} combines approximation techniques, a generalization of DPLL(T), and judicious SMT reasoning to \emph{lift} sound and complete linear-constraint \ac{OLNNV} tools to efficient, \emph{sound and complete} polynomial constraint verification.
    \chFourFancyName{} can exhaustively characterize unsafe state space regions (useful for retraining or the generation of fallback controllers).
    \item We implement \chFourFancyName{} for \relu{} \acp{NN} in the tool \ac{SNNT} and demonstrate our approach on three case studies from adaptive cruise control, airborne collision avoidance ACAS~X and steering under uncertainty.
    We show that, unlike \ac{SNNT}, State-of-the-Art \ac{CLNNV} tools cannot provide infinite-time guarantees due to approximation errors.
\end{itemize}

\begin{wrapfigure}[10]{r}{0.35\textwidth}
    \centering
    \resizebox{!}{0.9cm}{
    \begin{tikzpicture}%
    \pic at (0,0) {car={}};
    
    \pic at (5.,0) {car={}};
    
    \draw [<->] (2.1,0.3) --node[anchor=south, pos=0.5] {$\mathlarger{\mathlarger{\mathlarger{\accRelPos' = \accRelVel}}}$}node[anchor=north, pos=0.5] {$\mathlarger{\mathlarger{\mathlarger{\accRelVel' = \accRelAcc}}}$} (4.9,0.3);

    \draw [->] (7.1,0.3) --node[anchor=south, pos=0.5,yshift=0.1cm] {$\mathlarger{\mathlarger{\mathlarger{\accVelConst}}}$} (8.05,0.3);
    
    \end{tikzpicture}
    }
    \caption{Adaptive Cruise Control: The front-car (right) drives with constant $\accVelConst$. The ego-car approaches with relative velocity $\accRelVel$ (controlled via $\accRelAcc$) from $\accRelPos$.}
    \label{fig:running:acc_idea}
\end{wrapfigure}
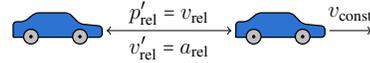

\paragraph{Running Example.}
\looseness=-1
The common NNCS safety benchmark of \acl{ACC}~\cite{Brosowsky2021,Fulton2018,ARCH21:ARCH_COMP21_Category_Report_Artificial} will serve as running example to demonstrate the introduced concepts.
Consider an ego-car following a front-car on a 1-D lane as shown in \Cref{fig:running:acc_idea}.
The front-car drives with constant velocity $\accVelConst$ while the ego-car (at position $\accRelPos$ behind the front-car) approaches with arbitrary initial (relative) velocity $\accRelVel$ which is adjusted through the ego-car's acceleration $\accRelAcc$.
The primary objective is to ensure the ego-car never crashes into the front car (i.e. $\accRelPos>0$), however there may be secondary objectives (e.g. energy efficiency) which are learned by an \ac{NNCS}.
\noindent
We demonstrate how a nondeterministic, high-level acceleration strategy (i.e. a safe envelope) can be modeled and verified in \ac{dL} (\Cref{sec:background}), 
how \chThreeFancyName{} derives NN properties (\Cref{sec:safe_controller_implementations}) and how such polynomial properties can be verified on a given \ac{NN} (\Cref{sec:nnnv}).
No techniques are specific to the running example, but all are applicable to a wide range of \acp{NNCS}---as demonstrated by our evaluation (\Cref{sec:evaluation}).

\section{Background}
\label{sec:background}
\looseness=-1
\begin{wraptable}[11]{r}{0.4\textwidth}%
\vspace*{-1em}
\begin{tabularx}{\textwidth}{l|X}
    Program & Semantics\\\hline
    $x \coloneqq e$ & Assign term $e$ to $x$ \\
    $x \coloneqq *$ & Nondet. assign to $x$ \\
    $?Q$ & Test of formula $Q$\\
    $x'=t\&Q$ & Evolve $x$ along the diff. equation within $Q$\\
    $\alpha \cup \beta$ & Nondet. choice\\
    $\alpha;\beta$ & Sequential composition\\
    $\left(\alpha\right)^*$ & Nondet. loop
\end{tabularx}
\caption{Program primitives of \ac{dL}}
\label{tab:dl_program_semantics}
\end{wraptable}
\looseness=-1
We review \ac{dL}, \acp{NN} and \ac{NN} verification.
$\formulaset{\reals}$ (resp. $\formulaset{\linreals}$) is the set of polynomial (resp. linear) real arithmetic first-order logic formulas. 
$\formulaset{\noetherianReals}$ extends $\formulaset{\reals}$ with \emph{Noetherian functions}~\cite{Platzer2020} $h_1,\dots,h_r$.
A \emph{Noetherian chain} is a sequence of real analytic functions $h_1,\dots,h_q$ s.t. all partial derivatives of all $h_j$ can be written as a polynomial $\frac{\partial h_j\mleft(y\mright)}{\partial y_i}\mleft(y\mright)=p_{ij}\mleft(y,h_1\mleft(y\mright),\dots,h_q\mleft(y\mright)\mright)$.
\emph{Noetherian functions} are representable as a polynomial over functions in a Noetherian chain.
Most %
activation functions used in \acp{NN} are Noetherian
(see \Cref{subsec:activation_table})
Atoms of a formula $\zeta$ are denoted as $\atomsof{\zeta}$ and its variables as $\varsofformula{\zeta}$.

\subsection{Differential Dynamic Logic}
\label{subsec:background:dl}
\looseness=-1
\Acf{dL}~\cite{Platzer2017a,Platzer2020,DBLP:conf/lics/Platzer12a,DBLP:journals/jar/Platzer08} is a first-order multi-modal logic in which the modality is parameterized with a \emph{hybrid program} describing a (discrete or continuous) state transition (see also \Cref{apx:additional_background}).
Thus, \ac{dL} formulas are evaluated in a state $\nu$ (if $\nu$ satisfies a formula $\psi$ we denote this as $\nu \vDash \psi$).
Hybrid Programs are constructed from the primitives in \Cref{tab:dl_program_semantics} and are first-class citizens of the logic (see example below).
\ac{dL} is tailored to the analysis of (time discrete and time continuous) hybrid systems and supports the analysis of differential equations.
Through its invariance reasoning capabilities, \ac{dL} allows us to prove the infinite-time safety of control envelopes w.r.t. a system's dynamics---even for cases where the dynamics' differential equations have no closed-form solution.
There is a large body of research on the verification of real-world control envelopes using \ac{dL} (e.g. ACAS X~\cite{Jeannin2017}). %
In \ac{dL}, the formula $\left[\alpha\right] \psi$ expresses that $\psi$ is always satisfied after the execution of $\alpha$ and $\left\langle \alpha\right\rangle \psi$ that there exists a state satisfying $\psi$ after the execution of $\alpha$.
\ac{dL} comes with a sound and relatively complete proof calculus~\cite{DBLP:journals/jar/Platzer08,Platzer2017a,Platzer2020} and an interactive theorem prover \acf{keymaerax}~\cite{Fulton2015}.

\paragraph{Running Example.}
\looseness=-1
We model our running example as a hybrid program in \ac{dL} with differential equations describing the evolution of $\accRelPos,\accRelVel,\accRelAcc$ along with a \emph{control envelope}, i.e. an abstract acceleration strategy, $\alpha_{\text{ctrl}}$ that runs at least every $T$ seconds while the overall system may run for arbitrarily many iterations (modeled by a nondeterministic loop).
Given suitable initial conditions (\text{accInit}), our objective is to prove the absence of collisions ($\accRelPos >0$).
This can be achieved by proving \Cref{eq:acc_contract}
where the place-holder $\alpha_{\text{ctrl}}$ determines the relative acceleration $-B \leq \accRelAcc \leq A$ ($A$ and $-B$ are resp. maximal acceleration/braking).
\begin{equation}
\label[formula]{eq:acc_contract}
\underbrace{\text{accInit}}_{\text{initial conditions}}
\implies
\Big[
\big( \underbrace{\alpha_{\text{ctrl}}}_{\text{controller}};c\coloneqq 0;\underbrace{(\accRelPos'=\accRelVel,\accRelVel'=-\accRelAcc,c'=1\,\&\,c \leq T)}_{\text{environment}} \big)^*
\Big] \underbrace{\accRelPos > 0}_{\text{safety constraint}}
\end{equation}
Our control envelope $\alpha_{\text{ctrl}}$ allows braking with $-B$ and an acceleration of $0$ or another value if the constraint $\text{accCtrl}_0$ or resp. $\text{accCtrl}_1$ is satisfied
(the concrete constraints are found in \Cref{apx:acc}):
\[
\alpha_{\text{ctrl}} ~\equiv~
\accRelAcc \coloneqq -B \cup
\left(\accRelAcc \coloneqq 0; ?\left(\text{accCtrl}_0\right)\right) \cup
\left(\accRelAcc \coloneqq *; ?\left(-B \leq \accRelAcc \leq A \land \text{accCtrl}_1\right)\right)
\]
The envelope is nondeterministic:
While always braking with $-B$ would be safe, an \ac{NNCS} may \emph{learn} to balance braking with secondary objectives (e.g. minimal acceleration or not falling behind).
A proof for \Cref{eq:acc_contract} in \ac{keymaerax} uses the loop invariant \text{accInv} (see \Cref{apx:acc}).
Automation of \ac{dL} proofs as well as control envelope and invariant synthesis is discussed in the literature~\cite{Platzer2017a,Platzer2020,SogokonMTCP22,DBLP:conf/tacas/KabraLMP24}.

\paragraph{ModelPlex.}
\looseness=-1
Many \ac{CPS} safety properties can be formulated through a \ac{dL} formula $\phi \implies \left[ \left( \alpha_{\text{ctrl}}; \alpha_{\text{plant}} \right)^* \right] \psi$ where $\phi$ describes initial conditions, and $\psi$ describes the safety criterion 
to be guaranteed when following the control-plant loop.
\Ac{modelplex} shielding~\cite{Mitsch2016} allows the synthesis of correct-by-construction \emph{controller monitor formulas} $\controllerFormula$ that ensure an implementation's runtime behavior matches the envelope $\alpha_{\text{ctrl}}$.
Interpreting an implementation's action as a state transition and denoting the old state's variables as $x_i$ and the new state's variables as $x_i^+$, $\controllerFormula$ tells us which combinations of $x_i$ and $x_i^+$ (i.e. which state transitions) are admissible w.r.t. $\alpha_{\text{ctrl}}$ (see \Cref{def:correct_monitor} in \Cref{apx:additional_background}).

\paragraph{Running Example.}
We can apply \ac{modelplex} on the proven contract in \Cref{eq:acc_contract} to synthesize a monitor for $\alpha_{\text{ctrl}}$.
For this simple scenario, the resulting controller monitor formula \footnotemark{} tells us what new acceleration value $\accRelAcc^+$ may be chosen given the current values of $\accRelPos,\accRelVel$:
\begin{equation}
\label[formula]{eq:acc_controller_formula}
\text{accCtrlFml}
\equiv
\accRelAcc^+=B \lor \left(\accRelAcc^+=0 \land \text{accCtrl}_0^+ \right) \lor \left(-B \leq \accRelAcc^+\leq A \land \text{accCtrl}_1^+ \right).
\end{equation}
Here, $\text{accCtrl}_i^+$ is the constraint $\text{accCtrl}_i$ with $\accRelAcc$ replaced by $\accRelAcc^+$.
\footnotetext{The formula furthermore keeps the new $\accRelPos^+$ and $\accRelVel^+$ values unchanged ($\accRelPos=\accRelPos^+ \land \accRelVel = \accRelVel^+$), elided.}
Given an action of a concrete controller implementation that changes $\accRelAcc$ to $\accRelAcc^+$,
\Cref{eq:acc_controller_formula} tells us if this action is in accordance with  the strategy modeled by $\alpha_{\text{ctrl}}$, i.e. whether we have a proof of safety for the given state transition.

\subsection{Neural Network Verification}
\label{subsec:background:nnv}
\looseness=-1
This work focuses on feed-forward neural networks typically encountered in \acp{NNCS}.
The behavior of an \ac{NN} with input dimension $I \in \mathbb{N}$ and output dimension $O \in \mathbb{N}$ can be summarized as a function $g: \reals^I \to \reals^O$.
The white-box behavior is described by a sequence of $L \in \mathbb{N}$ hidden layers with dimensions $n^{(k)}$ that iteratively transform an input vector $x^{(0)} \in \reals^I$ into an output vector $x^{(L)}\in\reals^O$.
The computation of layer $k$ is given by 
$
    x^{(k+1)} = f^{(k)}\left( W^{(k)} x^{(k)} + b^{(k)} \right)
$,
i.e. an affine transformation (with $\formulaset{\noetherianReals}$ representable numbers) followed by a nonlinear activation function $f^{(k)}$. %
We distinguish different classes of \acp{NN}. %
To this end, we decompose the activation functions $f^{(k)}$ as $f^{(k)}\mleft(x\mright) = \sum_{i=1}^s \mathds{1}_{q_i}\mleft(x\mright)f_i\mleft(x\mright)$ where $f_i$ are functions, $q_i$ are formulas over $n^{(k)}$ variables and $\mathds{1}_{q_i}\mleft(x\mright)$ is $1$ iff $q_i\mleft(x\mright)$ is true and $0$ otherwise.
\Cref{tab:nn_classes} summarizes which results are applicable to which \ac{NN} class.
\begin{table}[t]
    \centering
    \begin{tabular}{l|l|l|l|c|c}
        \ac{NN} class & All $f_i$ & All $q_i$ in & Applicable & Decidable & Example\\\hline
        piece-wise Noetherian & Noetherian &  $\formulaset{\noetherianReals}$ & \Cref{sec:safe_controller_implementations} & & Sigmoid \\\hline
        piece-wise Polynomial & Polynomial & $\formulaset{\reals}$ & \Cref{sec:safe_controller_implementations}  & $\checkmark$ & $x^2$\\\hline
        piece-wise Linear & Linear & $\formulaset{\linreals}$ & \Cref{sec:safe_controller_implementations,sec:nnnv}  & $\checkmark$ & MaxPool \\\hline
        ReLU & \multicolumn{2}{l|}{$f^{(k)}\left(x\right)=\max\left(0,x\right)$} & \Cref{sec:safe_controller_implementations,sec:nnnv,sec:evaluation}  & $\checkmark$ & $\relu$\\
    \end{tabular}
    \caption{Applicability of our results on \ac{NNCS} safety and decidability of the safety verification problem: Each class is a subset of its predecessor in the table.}
    \label{tab:nn_classes}
\end{table}
Each class is a subset of the previous class, i.e. our theory (\Cref{sec:safe_controller_implementations}) is widely applicable while our implementation (\Cref{sec:evaluation}) focuses on the most common \acp{NN}.
\Ac{OLNNV} tools  %
analyze \acp{NN} in order to verify properties on input-output relations.
Their common functionality is reflected in the VNNLIB standard~\cite{VNNLIB,DBLP:journals/sttt/BrixMBJL23}.
Off-the-shelf tools are limited to linear, normalized queries (\Cref{def:olnnv_interface}).
To address this challenge, we present a lifting procedure for the verification of generic (i.e. nonlinear and not normalized) \ac{OLNNV} queries over polynomial real arithmetic (\Cref{sec:nnnv}).
\begin{definition}[Open-Loop NNV Query]
\label{def:olnnv_interface}
An \emph{\ac{OLNNV} query} consists of a formula $p \in \formulaset{\reals}$ over free input variables $Z=\left\{z_1,\dots,z_I\right\}$ and output variables $x_1^+,\dots,x_O^+$.
We call $p$ \emph{normalized} iff 
$p$ is a conjunction of some input constraints and a disjunctive normal form over mixed/output constraints,
i.e. it has the structure
$
\bigwedge_j p_{1,j}\mleft(z_1,\dots,z_I\mright) \land \bigvee_{i\geq 2} \bigwedge_j p_{i,j}\mleft(z_1,\dots,z_I,x_1^+,\dots,x_O^+\mright),
$
where all $p_{i,j}$ are atomic real arithmetic formulas
and all $p_{1,j}$ only contain the free variables from $Z$.
We call a query \emph{linear} iff $p \in \formulaset{\linreals}$ and call it \emph{nonlinear} otherwise.
\end{definition}

\section{VerSAILLE: Verifiably Safe AI via Logically Linked Envelopes}
\label{sec:safe_controller_implementations}
\begin{textAtEnd}
\subsection{Proofs for \Cref{sec:safe_controller_implementations}}
\label{apx:proofs:versaille}
\looseness=-1
This subsection proves the soundness of the approach outlined in \Cref{sec:safe_controller_implementations}.
This result is achieved by proving that a concrete \ac{NNCS} refines~\cite{Loos2016,prebet2024uniform} an abstract hybrid program.
The approach can either be applied by first proving safety for a suitable \ac{dL} model or by reusing results from the \ac{dL} literature 
(both demonstrated  \Cref{sec:evaluation,apx:evaluation}).
In our proofs for \Cref{sec:safe_controller_implementations} we show a slightly more general version of the result from \Cref{thm:overview:soundness} (see \Cref{thm:overview:validtiy_safety}).
To this end, we formally define a controller description as follows:
\begin{definition}[Controller Description]
\label{def:overview:controller_implementation}
Let $\alpha_{\text{ctrl}}$ be some hybrid program with free variables $\freevars{\alpha_{\text{ctrl}}}=\left\{z_1,\dots,z_m\right\}$ and bound variables $\boundedvars{\alpha_{\text{ctrl}}}=\left\{x_1,\dots,x_n\right\}$,
which overlap if a variable is both read and written to.
A \emph{controller description} $\controllerDescriptionFml \in \formulaset{\noetherianReals}$ for $\alpha_{\text{ctrl}}$ is a formula with free variables $\freevars{\alpha_{\text{ctrl}}} \cup \left\{x^+ \mid x \in \boundedvars{\alpha_{\text{ctrl}}}\right\}$ such that the following formula is valid:
$
\forall\,z_1 \dots z_m\ \exists\,x_1^+ \dots x_n^+\ \controllerDescriptionFml
$.
\end{definition}
Based on Controller Descriptions we can then show that such controller descriptions exist for all piece-wise Noetherian \acp{NN}:
\begin{lemma}[Existence of $\nnControllerDescriptionFml$]
\label{lemma:overview:chi_g_controller_description}
Let $g:\mathbb{R}^I \to \mathbb{R}^O$ be a piece-wise Noetherian \ac{NN}. 
There exists a \emph{controller description} $\nnControllerDescriptionFml\in\formulaset{\noetherianReals}$
with \emph{input variables} $z_1,\dots,z_I$
and \emph{output variables} $x_1^+,\dots,x_O^+$
s.t.
$
\big(\nu\mleft(x_1^+\mright),\dots,\nu\mleft(x_O^+\mright)\big)
=
g\big(\nu\mleft(z_1\mright),\dots,\nu\mleft(z_I\mright)\big)
$
iff
$\nu \vDash \nnControllerDescriptionFml$,
i.e. $\nnControllerDescriptionFml$'s satisfying assignments correspond exactly to $g$'s in-out relation.
\end{lemma}
\begin{proof}
Previous work showed how to encode piece-wise linear \acp{NN} through real arithmetic SMT formulas (see e.g.~\cite{Papusha2020,Eleftheriadis2022}).
Each output dimension of an affine transformation can be directly encoded as a real arithmetic term.
For a given output-dimension of a piece-wise Noetherian activation function we have to encode a term $\sum_{i=1}^s \mathds{1}_{q_i}\mleft(x\mright) f_i\mleft(x\mright)$ with Noetherian functions $f_i$ and predicates $q_i$ over real arithmetic with Noetherian functions.
To this end, we can introduce fresh variables $v_1,\dots,v_s$ where we assert the following formula for each $v_i$:
\[
\left(q_i \land v_i = f_i\mleft(x\mright)\right) \lor \left(\neg q_i \land v_i = 0\right).
\]
The activation function's result then is the sum $\sum_{i=1}^s v_i$.
By existentially quantifying all intermediate variables of such encodings, we obtain a real arithmetic formula that only contains input and output variables and satisfies the requirements of the Lemma.

If we assign $x_1^+ \dots x_O^+$ to the values provided by $g\mleft(z_1,\dots,z_I\mright)$ for a given $z_1\dots z_I$, the formula $\nnControllerDescriptionFml$ is satisfied.
Therefore, $\forall z_1\dots z_I \exists x_1^+ \dots x_O^+ \nnControllerDescriptionFml.$ is valid.
\end{proof}
When replacing $\alpha_{ctrl}$ by an \ac{NN} $g$,
free and bound variables of $\alpha_{ctrl}$ must resp. match to input and output variables of $g$.
Based of a description $\nnControllerDescriptionFml$, we then construct a hybrid program that behaves as described by $\nnControllerDescriptionFml$:
\looseness=-1
We now formalize the idea of modeling a given \ac{NN} $g$ through a hybrid program which behaves identically to $g$.%
We show that such \emph{nondeterministic mirrors exist} for all piece-wise Noetherian \ac{NN} $g$:
\begin{definition}[Nondeterministic Mirror for $\nnControllerDescriptionFml$]
\label{def:overview:nondet_mirror}
\looseness=-1
Let $\alpha_{\text{ctrl}}$ be some hybrid program with bound variables $\boundedvars{\alpha_{\text{ctrl}}} = \left\{x_1, \dots, x_n\right\}$.
For a controller description $\nnControllerDescriptionFml$ with variables matching to $\alpha_{\text{ctrl}}$, $\nnControllerDescriptionFml$'s \emph{nondeterministic mirror} $\alpha_{\text{refl}}$ is defined as:\\
$
    \alpha_{\text{refl}}\mleft(\nnControllerDescriptionFml\mright) ~\equiv~ \left(x_1^+ \coloneqq *;\dots;x_n^+ \coloneqq *; ?\mleft(\nnControllerDescriptionFml\mright); x_1 \coloneqq x_1^+;\dots;x_n\coloneqq x_n^+\right)
$
\end{definition}
\begin{lemma}[Existence of $\nnProgramFml$]
\label[lemma]{lemma:nn_program_existence}
\looseness=-1
For any piece-wise Noetherian \ac{NN} $g:\reals^I \to \reals^O$ there exists a nondeterministic mirror $\nnProgramFml$ that behaves identically to $g$. \\
Formally, $\nnProgramFml$ only has free variables $\overline{z}$ and bound variables $\overline{x}$ and for any state transition $\left(\nu,\mu\right) \in \dLinterpretationapp{\nnProgramFml}$: %
$\mu\left(\overline{x}\right) = g\left(\nu\left(\overline{z}\right)\right)$ 
($\overline{x}$ and $\overline{z}$ vectors of dimension $I$ and $O$)
\end{lemma}
\begin{proof}
\label{proof:nn_program_existence}
Based on \Cref{lemma:overview:chi_g_controller_description} we can construct a controller description $\kappa_g \in \formulaset{\noetherianReals}$ for $g$ which we can turn into a hybrid program through the nondeterministic mirror $\alpha_{\text{refl}}\left(\kappa_g\right)$.
\end{proof}
Similarly to the more general notion of a controller description, \Cref{thm:overview:validtiy_safety} also permits a slightly more general version of a state space restriction instead of an inductive invariant.
Formally, this notion is described as a state reachability formula:
\begin{definition}[State Reachability Formula]
\label{def:complete_state_space_constraint}
A \emph{state reachability formula} $\stateFormula$ with free variables $z_1,\dots,z_m$ is \emph{complete} for the hybrid program $\left(\alpha_{\text{ctrl}};\alpha_{\text{plant}}\right)^*$ with free variables $z_1,\dots,z_m$ and initial state $\phi$ iff the following \ac{dL} formula is valid where $\left(\stateFormula\right)_{z_1 \ldots z_m}^{z_1^+ \ldots z_m^+}$ represents $\stateFormula$ with $z_i^+$ substituted for $z_i$ for all $1 \leq i \leq m$:
\begin{equation}
\label[formula]{eq:state_reachability_requirement}
\completeStateSpaceFormula.
\end{equation}
\end{definition}%
There is usually an overlap between free and bound variables, i.e. $z_1,\dots,z_m$ may contain variables later modified by the hybrid program.
Our definition requires that for any program starting in a state satisfying $\phi$, formula $\stateFormula$ is satisfied in all terminating states.
$\stateFormula$ thus overapproximates the program's reachable states.
In particular, inductive invariants %
(i.e. for $\phi \implies \left[\alpha^*\right] \psi$ a formula $\zeta$ s.t. $\phi \implies \zeta$ and  $\zeta \implies \left[\alpha\right] \zeta$)
are state reachability formulas:
\begin{lemma}[Inductive Invariants are State Reachability Formulas]
\looseness=-1
If $\zeta$ is an inductive invariant of $\phi \implies \left[\left(\alpha_{\text{ctrl}};\alpha_{\text{plant}}\right)^*\right] \psi$, $\zeta$ is a state reachability formula.
\end{lemma}%
\begin{proof}%
We begin by recalling the requirement for $\zeta$ to be a state reachability formula:
\[
\completeStateSpaceFormula{}.
\]
Let $\zeta$ be an inductive invariant for some contract of the form given above.
First, consider that for any state satisfying the left side of the formula above it holds that there exists some $k$ such that $\phi \land \langle \left(\alpha_{\text{ctrl}};\alpha_{\text{plant}}\right)^k \rangle \bigwedge_{i=1}^m z_i = z_i^+$ is satisfied by the same state (this follows from the semantics of loops in \ac{dL}).
If we can prove that any such state also satisfies $\left(\stateFormula\right)_{z_1,\dots,z_m}^{z_1^+,\dots,z_m^+}$ we obtain that $\zeta$ is a state reachability formula.
We proceed by induction:
First, consider $k=0$ in this case $z_i$ has the same value as $z_i^+$ for all $i$.
The formula then boils down to $\phi \implies \zeta$. This formula is guaranteed to be valid by the first requirement of inductive invariants.
Next, we now assume that we already proved that $\zeta$ holds for $k$ loop iterations and show it for $k+1$.
Since we assume some state $u_{k+1}$ that satisfies $\phi \land \langle \left(\alpha_{\text{ctrl}};\alpha_{\text{plant}}\right)^{k+1} \rangle \bigwedge_{i=1}^m z_i = z_i^+$, there also has to be some state $u_k$ satisfying $\phi \land \langle \left(\alpha_{\text{ctrl}};\alpha_{\text{plant}}\right)^{k} \rangle \bigwedge_{i=1}^m z_i = z_i^+$.
However, we already know for $u_k$ that it satisfies $\left(\stateFormula\right)_{z_1,\dots,z_m}^{z_1^+,\dots,z_m^+}$.
Since $u_{k+1}$ is reachable from $u_k$ through the execution of $\alpha_{\text{ctrl}};\alpha_{\text{plant}}$ we know that $u_{k+1}$ satisfies $\left(\stateFormula\right)_{z_1,\dots,z_m}^{z_1^+,\dots,z_m^+}$ (this corresponds to the property $\zeta \implies \left[\alpha\right]\zeta$ of inductive invariants).
\end{proof}%
\looseness=-1%
As $\nnProgramFml$ and $g$ mirror each other, we can reason about them interchangeably.
Our objective is now to prove that $\alpha_g$ is a refinement of $\alpha_{\text{ctrl}}$.
To this end, we use the shielding technique \ac{modelplex}~\cite{Mitsch2016} to
automatically generate a correct-by-construction controller monitor $\controllerFormula$ for $\alpha_{\text{ctrl}}$.
The formula $\controllerFormula$ then describes what behavior for $\nnProgramFml$ is acceptable so that $\nnProgramFml$ still represents a refinement of $\alpha_{\text{ctrl}}$.
As seen in \Cref{sec:safe_controller_implementations}, we do \emph{not} require that $\nnProgramFml$ adheres to $\controllerFormula$ on all states, but only on reachable states. %
For efficiency we therefore allow limiting the analyzed state space to an inductive invariant $\stateFormula$ (i.e. for $\phi \implies \left[\alpha^*\right] \psi$ a formula $\stateFormula$ s.t. $\phi \implies \stateFormula$ and  $\stateFormula \implies \left[\alpha\right] \stateFormula$).
Despite the infinite-time horizon, the practical use of our approach often faces implementations with a limited value range for inputs and outputs (e.g.,
$\accRelVel$ within the ego-car's physical capabilities).
Only by exploiting these ranges, is it possible to prove safety for \acp{NN} that were only trained on a particular value range. %
To this end, we allow specifying value \emph{ranges} (i.e. intervals) for variables.
We define the range formula $R \equiv \bigwedge_{v \in \varsofformula{P}} \underline{R}\mleft(v\mright) \leq v \leq \overline{R}\mleft(v\mright)$ for lower and upper bounds $\underline{R}$ and $\overline{R}$.
Using \(R\), we specialize a contract to the implementation specifics by adding a range check to $\alpha_\text{plant}$.
The safety results for the original contract can be reused:
\begin{lemma}[Range Restriction]
\label{lem:overview:introduce_bounds}
Let $\hybridsystemscontract{}$ be a valid \ac{dL} formula.
Then the formula
$C_2\equiv\left(\phi \land R \implies \left[\left(\alpha_{\text{ctrl}};\left(\alpha_{\text{plant}};?\mleft(R\mright)\right)\right)^*\right]  \psi\right)$
with ranges $R$ is valid and \(R\) is an invariant for $C_2$.
\end{lemma}
\begin{proof}
We use the notation $C_1\equiv\left(\hybridsystemscontract{}\right)$.
We begin by showing that the validity of $C_1$ implies the validity of $C_2$.
Intuitively, this follows from the fact that the introduced check $?\mleft(R\mright)$ only takes away states.
Let $I$ be a loop invariant such that $\phi \implies I$, $I \implies \psi$ and $I \implies \left[\alpha_{\text{ctrl}};\alpha_{\text{plant}}\right] I$ (assumed due to the validity of $C_1$).
Then clearly, it also holds that $\phi \land R \implies I$.
Furthermore, $I \implies \left[\alpha_{\text{ctrl}};\alpha_{\text{plant}};?\mleft(R\mright)\right]I$ can be reduced to $I \implies \left[\alpha_{\text{ctrl}};\alpha_{\text{plant}}\right]\left(R \implies I\right)$ which we can shown through the monotonicity rule of \ac{dL}.
Since we already know that $I \implies \psi$, it follows that $C_2$ is valid, because $I$ is a loop invariant.
\end{proof}
Including $R$ into $\stateFormula$ allows us to exploit the range limits for the analysis of $\nnProgramFml$.
Our objective is to use \ac{OLNNV} techniques to check whether $g$ (and therefore $\nnProgramFml$) satisfies the specification synthesized by \ac{modelplex}.
To this end, we use a \emph{nonlinear \ac{NN} verifier} to prove safety of our \ac{NNCS}:
\begin{theorem}[Safety Criterion]
\label{thm:overview:validtiy_safety}
Let $\controllerFormula$ and $\stateFormula$ be controller and state reachability formulas for
a valid \ac{dL} contract $C\equiv\left(\hybridsystemscontract{}\right)$.
For any controller description $\controllerDescriptionFml$, if
\begin{equation}
\label[formula]{eq:safety_property}
\realArithmeticSafetyGuarantee
\end{equation}
is valid, then the following \acs{dL} formula is valid as well:
\begin{equation}
    \label[formula]{eq:safety_guarantee}
\safetyGuarantee
\end{equation}
\end{theorem}
\begin{proof}
Assume the validity of 
\[
\realArithmeticSafetyGuarantee.
\]
Let $v \in \dLstatespace$ be some arbitrary state.
We need to show that any such $v$ satisfies \Cref{eq:safety_guarantee}:
\[
\safetyGuarantee.
\]
To this end, assume $v \vDash \phi$, we prove that $\psi$ as well as $\stateFormula$ is upheld after any number of loop iterations by induction on the number $n$ of loop iterations.\\
\bigskip\\
\textit{Base Case:} $n=0$\\
In this case, the only state we need to consider is $v$ since there were no loop iterations.
We know through the validity of $C$ that $\phi \implies \psi$.
Thus, $v \vDash \psi$.
Furthermore, we recall the requirement of a state reachability formula:\[
\completeStateSpaceFormula{}.
\]
By extending $v$ such that all $z_i^+$ have the same values as the corresponding $z_i$, we get a state that satisfies this formula.
Consequently, $v \vDash \stateFormula$.
\bigskip\\
\textit{Inductive Case:} $n \to \left(n+1\right)$\\
In the induction case, we know that for all %
\[
\left(v,\tilde{v}_0\right) \in \dLinterpretationapp{\left(\safetyGuaranteeProgramNoLoop\right)^n}
\]
it holds that
$\tilde{v}_0 \vDash \psi$ and $\tilde{v}_0 \vDash \stateFormula$.\\
We must now prove the induction property for any state reachable from $\tilde{v_0}$ through execution of the program $\left(\safetyGuaranteeProgramNoLoop\right)$.
For any $\tilde{v}_1\in\dLstatespace$ such that $\left(\tilde{v}_0,\tilde{v}_1\right) \in \dLinterpretationapp{x_1^+\coloneqq*;\dots;x_n^+\coloneqq*;?\mleft(\controllerDescriptionFml\mright);}$ (by the definition of $\controllerDescriptionFml$ we know that such a state exists) we know that $\tilde{v}_1 \vDash \controllerDescriptionFml$.
According to the coincidence lemma~\cite[Lemma 3]{Mitsch2016}, since $\stateFormula$ does not concern the $x^+$ variables, it is then true that
\[
\tilde{v}_1 \vDash \stateFormula \land \controllerDescriptionFml.
\]
Through the validity of \Cref{eq:safety_property} assumed at the beginning, we then know that it must be the case that $\tilde{v}_1 \vDash \controllerFormula$.
More specifically, this means that for any $\tilde{v}_2\in\dLstatespace$ with
\[
\left(\tilde{v}_1,\tilde{v}_2\right) \in \dLinterpretationapp{x_1\coloneqq x^+_1;\dots;x_n\coloneqq x_n^+};
\]
it holds that $\left(\tilde{v}_0,\tilde{v}_2\right) \vDash \controllerFormula$.
By definition this implies that $\left(\tilde{v}_0,\tilde{v}_2\right)\in\dLinterpretationapp{\alpha_{\text{ctrl}}}$.%

In summary, this means that for any $\left(\tilde{v}_0,\tilde{v}_2\right)\in\dLinterpretationapp{\alpha_{\text{refl}}\mleft(\controllerDescriptionFml\mright)}$ it holds that $\left(\tilde{v}_0,\tilde{v}_2\right)\in\dLinterpretationapp{\alpha_{\text{ctrl}}}$.%

Through the semantics of program composition in hybrid programs it follows that subsequently for any $\tilde{v}_3 \in \dLstatespace$ with $\left(\tilde{v}_0,\tilde{v}_3\right)\in\dLinterpretationapp{\alpha_{\text{refl}}\mleft(\controllerDescriptionFml\mright);\alpha_{\text{plant}}}$ it holds that $\left(\tilde{v}_0,\tilde{v}_3\right)\in\dLinterpretationapp{\alpha_{\text{ctrl}};\alpha_{\text{plant}}}$.

We also know that $\left(v,\tilde{v}_0\right) \in \dLinterpretationapp{\defaultcontroller}$ and that $\left(\tilde{v}_0,\tilde{v}_3\right) \in \dLinterpretationapp{\alpha_{\text{ctrl}};\alpha_{\text{plant}}}$.
Since this implies $\left(v,\tilde{v}_3\right) \in \dLinterpretationapp{\defaultcontroller}$, i.e. there is a trace of states from $v$ to $\tilde{v}_3$, and since we already know that $v \vDash \phi$, $\left(v,\tilde{v}_3\right)$ satisfy the right side of \Cref{eq:state_reachability_requirement}.
Since \Cref{eq:state_reachability_requirement} must be valid we get that $\tilde{v}_3 \vDash \stateFormula$.
Consequently, we know through the validity of $C$ that:
\[
\tilde{v}_3 \vDash \psi \land \stateFormula.
\]
This concludes the induction proof and thereby also the proof of \Cref{thm:overview:validtiy_safety}
\end{proof}%

\begin{definition}[Nonlinear Neural Network Verifier]
\label{def:nonlinear_olnnv}
A \emph{nonlinear neural network verifier} accepts as input a piece-wise Noetherian \ac{NN} $g$ and nonlinear \ac{OLNNV} query $p$ with free variables  $z_1,\dots,z_I,x_1^+,\dots,x_O^+$.
The tool must be \emph{sound}, i.e. if there is a $z \in \reals^I$ satisfying $p\mleft(x,g\mleft(z\mright)\mright)$ then the tool must return \texttt{sat}.
A tool that always returns \texttt{unsat} if no such $z \in \reals^I$ exists is called \emph{complete}.
\end{definition}%
\begin{lemma}[Soundness w.r.t Controller Descriptions]
\label{lemma:overview:soundness_controller_descriptions}
Let $\nnControllerDescriptionFml$ be a controller description for a piece-wise Noetherian \ac{NN} $g$.\\
Further, let $C\equiv\left(\hybridsystemscontract{}\right)$ be a contract with controller monitor $\controllerFormula\in\formulaset{\reals}$ and inductive invariant $\stateFormula\in\formulaset{\reals}$ where the free and bound variables respectively match $g$'s inputs and outputs.
If a sound Nonlinear Neural Network Verifier returns \texttt{unsat} for the query $p \equiv \left(\stateFormula \land \neg \controllerFormula\right)$ on $g$ then:
\begin{enumerate*}
    \item $\nnControllerDescriptionFml \land \stateFormula \implies \controllerFormula$ is valid;
    \item $\phi \implies \left[\left(\alpha_{\text{refl}}\left(\nnControllerDescriptionFml\right);\alpha_{\text{plant}}\right)^*\right] \psi$ is valid.
\end{enumerate*}
\end{lemma}%
\begin{proof}
Let all variables be defined as above.
We assume that the nonlinear \ac{NN} verifier did indeed return \texttt{unsat}.
By definition this means that there exists no $z \in \mathbb{R}^I$ such that $p\mleft(z,g\mleft(z\mright)\mright) = \truesym$.
Due to the formalization of $\nnControllerDescriptionFml$ (see \Cref{lemma:overview:chi_g_controller_description}), this means there exists no $z \in \mathbb{R}^I$ such that $\stateFormula \land \nnControllerDescriptionFml \land \neg \controllerFormula$.
Among all states consider now any state $v$ such that $v \nvDash \stateFormula \land \nnControllerDescriptionFml$.
In this case $v \vDash \stateFormula \land \nnControllerDescriptionFml \implies \controllerFormula$ vacuously.
Next, consider the other case, i.e. a state $v$ such that $v \vDash \stateFormula \land \nnControllerDescriptionFml$.
In this case it must hold that $v \nvDash \neg \controllerFormula$.
So $v \vDash \controllerFormula$.
Therefore, $v \vDash \stateFormula \land \nnControllerDescriptionFml \implies \controllerFormula$.
This means that \Cref{eq:safety_property} is satisfied by all states and, therefore, valid which proves the first claim.
\Cref{thm:overview:validtiy_safety} then implies the safety guarantee stated in \Cref{eq:safety_guarantee} for $\nnControllerDescriptionFml$ which proves the second claim.
\end{proof}%
While we lay the foundation for analyses on piece-wise Noetherian \acp{NN},
a subclass is decidable:
\begin{lemma}[Decidability for Polynomial Constraints]
\label{lemma:decidability}
Given a piece-wise polynomial \ac{NN} $g$, 
the problem of verifying $\left(\stateFormula \land \neg\controllerFormula\right) \in \formulaset{\reals}$ for $g$ is decidable.
\end{lemma}
\begin{proof}
\label{proof:decidability}
The problem of verifying $\left(\stateFormula \land \neg\controllerFormula\right) \in \formulaset{\reals}$ for $g$ is the same as proving the validity of the formula $\stateFormula \land \nnControllerDescriptionFml \implies \controllerFormula$ (see \Cref{lemma:overview:soundness_controller_descriptions}).
For piece-wise polynomial \ac{NN} this formula is in $\formulaset{\reals}$ and the validity problem is thus decidable.
\end{proof}
\end{textAtEnd}
\looseness=-1
We introduce \chThreeFancyName{}, our approach for the verification of \acp{NNCS} via \ac{dL} contracts.
The key idea of \chThreeFancyName{} are \emph{nondeterministic mirrors},
a mechanism that allows us to reflect a given \ac{NN} $g$ and reason within and outside of \ac{dL} simultaneously.
This allows us to instrument \ac{OLNNV} techniques to prove an \ac{NN} specification outside of \ac{dL} which \emph{implies} the safety of a corresponding (mirrored) \ac{dL} model describing the \ac{NNCS}.
Reconsider the \ac{ACC} example (\Cref{sec:introduction}) for which we synthesized a controller monitor formula in \Cref{sec:background}.
The remaining open question is the following:
\begin{quote}
If we replace the control envelope $\alpha_{\text{ctrl}}$ by a given piece-wise Noetherian \ac{NN} $g$, does the resulting system retain the same safety guarantees?
\end{quote}
\looseness=-1

\paragraph{Summary of \chThreeFancyName{}}
\looseness=-1
The input for \chThreeFancyName{} is a proven \ac{dL} contract.
Additionally, one may provide an inductive invariant $\stateFormula$ to simplify the subsequent state space analysis.
Using \ac{modelplex}'s synthesis of $\controllerFormula$, \chThreeFancyName{} constructs a nonlinear \ac{OLNNV} query.
If we verify this query on an \ac{NN} $g$, then the \ac{NNCS} where we \emph{substitute} the control envelope by $g$ retains the same safety guarantee.

\paragraph{Running Example}
\looseness=-1
Since one can only provide formal guarantees for something one can describe formally, we first need a semantics for what it means to substitute $\alpha_{\text{ctrl}}$ by $g$.
To this end, we formalize a given piece-wise Noetherian \ac{NN} $g$ as a
hybrid program $\nnProgramFml$ which we call the \emph{nondeterministic mirror} of $g$ (see \Cref{def:overview:nondet_mirror} and \Cref{lemma:nn_program_existence} in \Cref{apx:proofs}).
E.g. for \ac{ACC}, this program must have two free (i.e. read) variables $\accRelPos,\accRelVel$ and one bound (i.e. written) variable $\accRelAcc$
and must be designed in such a way that it exactly implements the \ac{NN} $g$.
Showing safety (see question above) is then equivalent to proving the following \ac{dL} formula where $\alpha_{\text{plant}}$ describes \ac{ACC}'s physical dynamics:
\begin{equation}
\label[formula]{eq:acc_guarantee}
\text{accInit} \implies \left[\left(\alpha_g;\alpha_{\text{plant}}\right)^*\right]\accRelPos>0.
\end{equation}
\looseness=-1
Since \acp{NN} do not lend themselves well to interactive analysis, an \emph{automatable} mechanism to prove \cref{eq:acc_guarantee} \emph{outside} the \ac{dL} calculus is desirable.
As discussed in \Cref{sec:background}, we can prove the safety ($\accRelPos>0$) of $\alpha_{\text{ctrl}}$ via \ac{dL}.
Thus, if we can show that all behavior of the nondeterministic mirror $\nnProgramFml$ is \emph{already} modeled by $\alpha_{\text{ctrl}}$, the safety guarantee carries over from the envelope to $\nnProgramFml$. %
To show this refinement relation~\cite{Loos2016,prebet2024uniform}, we instrument the controller monitor $\text{accCtrlFml}$ in \Cref{eq:acc_controller_formula}. %
We verify that the \ac{NN} $g$ satisfies the controller monitor formula $\text{accCtrlFml}$
(i.e. 
we show that $g$'s input-output relation satisfies \Cref{eq:acc_controller_formula}).
If this is the case, $\nnProgramFml$'s behavior is modeled by our envelope $\alpha_{\text{ctrl}}$.
In practice, it is unnecessary that the behavior of $\nnProgramFml$ is modeled by $\alpha_{\text{ctrl}}$ \emph{everywhere}
(e.g. we are not interested in states with $\accRelPos \leq 0$).
It suffices to consider all states within the inductive invariant $\text{accInv}$ of the envelope's system (\Cref{eq:acc_contract}) as those are precisely the states for which the guarantee on $\alpha_{\text{ctrl}}$ holds.
Thus, we can prove \Cref{eq:acc_guarantee} by showing that $g$ satisfies the following specification for all inputs:
\begin{equation}
\label[formula]{eq:acc_olnnv_property}
\text{accInv} \implies \text{accCtrlFml}
\end{equation}
\looseness=-1
\chThreeFancyName{} also allows to soundly constrain the system to value ranges where the \ac{NN} has been trained (\Cref{lem:overview:introduce_bounds}).
Verifying the queries generated by \chThreeFancyName{} with a (nonlinear) \ac{OLNNV} tool (\Cref{def:nonlinear_olnnv}) then implies safety of the \ac{NNCS} (full formalism see \Cref{apx:proofs:versaille}):
\begin{theoremE}[Soundness][end,text link=]%
\label{thm:overview:soundness}%
Let $g$ be a piece-wise Noetherian \ac{NN}.
Further, let $C\equiv\left(\hybridsystemscontract{}\right)$ be a valid contract with controller monitor $\controllerFormula\in\formulaset{\reals}$ and inductive invariant $\stateFormula\in\formulaset{\reals}$.
If a sound Nonlinear Neural Network Verifier returns \texttt{unsat} for the query $p \equiv \left(\stateFormula \land \neg \controllerFormula\right)$ on $g$ then
$\phi \implies \left[\left(\nnProgramFml;\alpha_{\text{plant}}\right)^*\right] \psi$ is valid
for the nondet. mirror $\nnProgramFml$.
\end{theoremE}%
\begin{proofE}
\label{proof:overview:soundness}%
The formula $\phi \implies \left[\left(\nnProgramFml;\alpha_{\text{plant}}\right)^*\right] \psi$ is equivalent to $\phi \implies \left[\left(\alpha_{\text{refl}}\left(\nnControllerDescriptionFml\right);\alpha_{\text{plant}}\right)^*\right] \psi$ as $\nnProgramFml$ behaves precisely like $\alpha_{\text{refl}}\left(\nnControllerDescriptionFml\right)$.
The result immediately follows from \Cref{lemma:overview:soundness_controller_descriptions}
\end{proofE}
If $g$ is piece-wise polynomial, \Cref{eq:acc_olnnv_property} is expressible in \formulaset{\reals} and therefore its verification decidable (\Cref{lemma:decidability} which is a special case of \cite{prebet2024uniform}).
In practice, we can be much more efficient than naively applying real arithmetic theory solvers by relying on \ac{OLNNV} technologies to check the negated property $\text{accInv} \land \neg \text{accCtrlFml}$:
If this property is unsatisfiable for a given \ac{NN} $g$, then 
\Cref{eq:acc_guarantee} is valid.
Off-the-shelf \ac{OLNNV} tools are unable to reason about \Cref{eq:acc_olnnv_property} due to its nonlinearities and the non-normalized formula structure.
Thus, our second contribution (\Cref{sec:nnnv})
lifts \ac{OLNNV} tools to the task of verifying
nonlinear, non-normalized queries.

\section{\chFourFancyName{}: Nonlinear Open-Loop NN Verification}%
\label{sec:nnnv}
\begin{textAtEnd}
    \subsection{Proofs for \Cref{sec:nnnv}}
\end{textAtEnd}
\newcommand{\normalizeAlg}[0]{\textsc{Mosaic}}%
\newcommand{\generalizeAlg}[0]{\textsc{Generalize}}%
\newcommand{\enumerateAlg}[0]{\textsc{Enum}}%
\newcommand{\overapproxAlg}[0]{\textsc{Linearize}}%
\newcommand{\filterAlg}[0]{\textsc{Filter}}%
\looseness=-1
Since \acp{NNCS} usually exhibit nonlinear physical behavior, the verification property $\stateFormula \land \neg\controllerFormula$ will be nonlinear as well.
$\stateFormula \land \neg\controllerFormula$ is also a formula of arbitrary structure and not a normalized \ac{OLNNV} query (\Cref{def:olnnv_interface}).
This is evident in the verification query for our running example (see \Cref{apx:acc}) or in \Cref{fig:ncubev} where the green and blue constraints (left) describe two independent nonlinear input regions ($\mathrm{case}_1$ and $\mathrm{case}_2$) and the red regions (right) correspond to the unsafe regions for $\mathrm{case}_1$.
In contrast to this, off-the-shelf \ac{OLNNV} tools (e.g. nnenum~\cite{bak2020improved} or Marabou~\cite{katz2019marabou}) only support the verification of linear, normalized \ac{OLNNV} queries on piece-wise linear \acp{NN}.
\chFourFancyName{} is a framework that allows us to \emph{lift} off-the-shelf \ac{OLNNV} tools for linear, normalized \ac{OLNNV} queries to polynomial queries of arbitrary logical structure.
This approach has the notable advantage that \chFourFancyName{}'s capabilities grow as \ac{OLNNV} technology advances while retaining completeness w.r.t. polynomial constraints.
An overview of the algorithm is given in \Cref{algorithm:nnnv:the_algorithm} and \Cref{fig:ncubev} which is explained throughout this section with further details in \Cref{apx:mosaic}.
\def\FunctionFOne(#1){-4.21*(#1)^0.5}%
\def\FunctionZero(#1){0}%
\def\FunctionG(#1){4.21*(-#1)^0.5-30}%
\def\FunctionGTwo(#1){-30}%
\def\FunctionUOne(#1){-(0.1*(#1-30))^2+30}%
\def\FunctionUTwo(#1){(0.4*(#1-15))^2+70}%

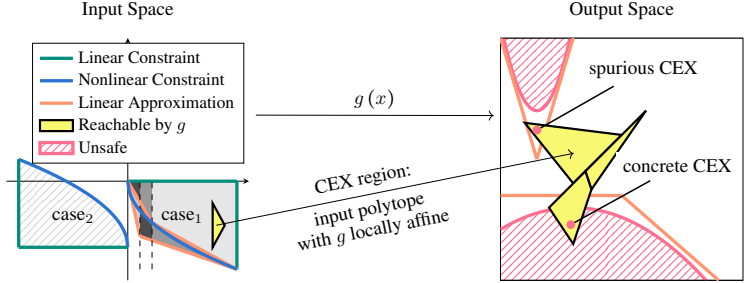
\begin{figure}[t]%
    \centering
    \resizebox{0.35\linewidth}{!}{
    \begin{tikzpicture}
    \begin{axis}[
            axis y line=center,
            axis x line=middle, 
            axis on top=false,
            xmin=-100,
            xmax=100,
            ymin=-45,
            ymax=65,
            xshift=-5cm,
            height=6.0cm,
            width=6.0cm,
            ticks=none,
            legend style={font=\small},
            legend cell align=left,
            title={Input Space}
        ]
        \addplot [name path = F2, domain=0:90, samples=100, mark=none, line width=0.6mm, c11] {\FunctionZero(x)};
        
        \draw [name path = F3, line width=0.6mm, c11] (90,-40) -- (90,0);
        \addlegendentry{Linear Constraint}
        \addlegendentry{Nonlinear Constraint}
        \addlegendentry{Linear Approximation}
        \addlegendentry{Reachable by $g$}
        \addlegendentry{Unsafe}
        \addlegendimage{line legend,line width=0.6mm,c15}
        \addlegendimage{line legend,line width=0.6mm,c3}
        \addlegendimage{area legend,line width=0.4mm,fill=c1}
        \addlegendimage{area legend,line width=0.6mm,pattern color=c4, pattern=north east lines,draw=c4}
        
        \draw [name path = FA1, line width=0.6mm, c3] (20,-18.82) -- (90,-40);
        \draw [name path = FA2, line width=0.6mm, c3] (0.0, 0.0) -- (20,-18.82);

        \draw [name path = FA3, line width=0.6mm, c3] (10,-25) -- (90,-40);
        \draw [name path = FA4, line width=0.6mm, c3] (0.0, 0.0) -- (10,-25);

        \draw [dashed,black] (10,-40) -- (10,0);
        \draw [dashed,black] (20,-40) -- (20,0);

        \node (case1) [left, black,anchor=south] at (axis cs: 45,-22) {$\text{case}_1$};

        \addplot [name path = F1, domain=0:90, samples=100, mark=none, line width=0.6mm, c15] {\FunctionFOne(x)};
        
        \addplot [name path = G1,domain=-90:0, samples=100, mark=none, line width=0.6mm, c15] {\FunctionG(x)};
        \addplot [name path = G2, domain=-90:0, samples=100, mark=none, line width=0.6mm, c11] {\FunctionGTwo(x)};
        \draw [name path = G3,line width=0.6mm, c11] (-90,-30) -- (-90,10);

        \addplot [pattern color=gray!30, pattern=north east lines] fill between [of = G1 and G2, soft clip={domain=-90:0}];
        
        \node [left, black,anchor=south] at (axis cs: -45,-22) {$\text{case}_2$};

        \addplot [black!10] fill between [of = FA1 and F2, soft clip={domain=20:90}];
        \addplot [black!40] fill between [of = FA1 and FA3, soft clip={domain=20:90}];
        \addplot [black!40] fill between [of = FA2 and F2, soft clip={domain=10:20}];
        \addplot [black!70] fill between [of = FA2 and FA3, soft clip={domain=10:20}];
        \addplot [black!20] fill between [of = FA2 and FA4, soft clip={domain=0:10}];
        \addplot [black!70] fill between [of = FA2 and F2, soft clip={domain=0:10}];

        \filldraw[line width=0.4mm,fill=c1,on layer=foreground] (70,-10) -- (80,-20) -- (70,-30) -- cycle; 

        \node (axis1End) at (axis cs:100,30) {};
        \node (counterexampleInput) at (axis cs:73,-20) {};
    \end{axis}
    
    \begin{axis}[
            axis on top=false,
            ticks=none,
            xmin=0,
            xmax=100,
            ymin=0,
            ymax=100,
            xshift=4cm,
            height=6.0cm,
            width=6.0cm,
            legend style={font=\small},
            legend cell align=left,
            title={Output Space}
        ]

        \addplot [name path = U1, domain=0:100, samples=100, mark=none, line width=0.6mm, c4] {\FunctionUOne(x)};
        \addplot[name path = UZero, domain=0:100, samples=100, mark=none, line width=0mm] {0};
        \addplot [pattern color=c4, pattern=north east lines] fill between [of = U1 and UZero, soft clip={domain=0:87}];

        \addplot [name path = U2, domain=0:100, samples=100, mark=none, line width=0.6mm, c4] {\FunctionUTwo(x)};
        \addplot[name path = UTop, domain=0:100, samples=100, mark=none, line width=0mm] {100};
        \addplot [pattern color=c4, pattern=north east lines] fill between [of = U2 and UTop, soft clip={domain=0:87}];

        \draw [name path = A1, line width=0.6mm, c3] (0,35) -- (50,35);
        \draw [name path = A2, line width=0.6mm, c3] (50,35) -- (90,0);

        \draw [name path = A3, line width=0.6mm, c3] (0,100) -- (15,50);
        \draw [name path = A4, line width=0.6mm, c3] (30,100) -- (15,50);

        \filldraw[line width=0.4mm,fill=c1,on layer=foreground] (10,65)  -- (50,60) -- (30,40) -- cycle; 
        \filldraw[line width=0.4mm,fill=c1,on layer=foreground] (35,45) -- (60,70) -- (37.5,37.5) -- cycle; 
        \filldraw[line width=0.4mm,fill=c1,on layer=foreground] (20,30) -- (35,45) -- (37.5,37.5) -- (30,15) -- cycle; 

        \node [fill, draw, circle, minimum width=4pt, inner sep=0pt, pin={[pin edge={black,thick},fill=white, outer sep=1pt,pin distance=30]45:{spurious CEX}},c4] at (15,62) {};

        \node [fill, draw, circle, minimum width=4pt, inner sep=0pt, pin={[pin edge={black,thick},fill=white, outer sep=1pt,pin distance=30]45:{concrete CEX}},c4] at (29,23) {};
        
        \node (axis2End) at (axis cs:0,68) {};
        \node (counterexampleOutput) at (axis cs:33,53) {};
    \end{axis}
    \draw[->] (axis1End) --node[anchor=south] {$g\left(x\right)$} (axis2End);
    \draw[->] (counterexampleInput.center) -- (counterexampleOutput) node[pos=0.42, above, sloped]{CEX region:} node[pos=0.42, below, sloped,align=center]{input polytope\\with $g$ locally affine};
    \end{tikzpicture}
    }
    \caption{Visualization of the nonlinear verification algorithm \chFourFancyName{} in \Cref{sec:nnnv}}
    \label{fig:ncubev}
\end{figure}

\paragraph[Mosaic]{\normalizeAlg{}.}
\looseness=-1
\begin{wrapfigure}[11]{r}{0.25\textwidth}
    \centering
    \resizebox{\linewidth}{!}{
    \begin{tikzpicture}
    \begin{axis}[
            axis y line=center,
            axis x line=middle, 
            axis on top=false,
            xmin=0,
            xmax=100,
            ymin=0,
            ymax=100,
            height=2.75cm,
            width=3.85cm,
            x label style={at={(axis description cs:0.5,-0.5), font=\small},anchor=north},
            y label style={at={(axis description cs:-0,.5), font=\small},rotate=90,anchor=south},
            xlabel={Input Space},
            ylabel={Output},
            legend style={font=\small},
            legend cell align=left,
            xtick={20,60,80},
            xticklabels={$A_1$,$A_2$,$A_3$},
            ymajorticks=false
        ]

        \draw[thick, black, fill=c4] (10,10) rectangle (30,30);

        \draw[thick, black, fill=c4] (50,10) rectangle (90,30);

        \draw[thick, black, fill=c4] (10,50) rectangle (30,60);

        \draw[thick, black, fill=c4] (10,80) rectangle (30,90);

        \draw[thick, black, fill=c4] (50,40) rectangle (70,50);

        \draw[thick, black, fill=c4] (50,70) rectangle (70,100);

        \draw[thick, black, fill=c4] (70,55) rectangle (90,65);

        \draw [dashed,black] (10,0) -- (10,100);
        \draw [dashed,black] (30,0) -- (30,100);
        \draw [dashed,black] (50,0) -- (50,100);
        \draw [dashed,black] (70,0) -- (70,100);
        \draw [dashed,black] (90,0) -- (90,100);
        
    \end{axis}
    \end{tikzpicture}
    }
    \caption{Enumeration of specification regions}
    \label{fig:dpllt_vs_mosaic}
\end{wrapfigure}
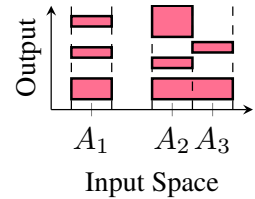
To tackle formulas with arbitrary logical structure we could use DNNV~\cite{Shriver2021} which implements a simple expansion algorithm or the standard formulation of DPLL(T).
Consider the specification in \Cref{fig:dpllt_vs_mosaic} where the x-axis represents the \ac{NN}'s inputs, the y-axis possible outputs, and the red regions are considered unsafe:
DPLL(T), or an expansion algorithm, would enumerate all 7 red regions individually and then invoke an \ac{OLNNV} tool for each.
However, this becomes prohibitively inefficient when used for reachability-based \ac{OLNNV} tools: Such tools will compute the reachable regions for $A_1$ three times, for $A_2$ three times, and for $A_3$ two times.
Instead, we propose $\normalizeAlg{}$ which first only enumerates \emph{input} regions (e.g. the regions $A_1$, $A_2$, and $A_3$) and then enumerates unsafe regions in the output space \emph{per input region}.
The unsafe regions are then aggregated into disjunctive normal form.
Hence, for the example in \Cref{fig:dpllt_vs_mosaic}, $\normalizeAlg{}$ would only yield \emph{three} normalized queries which are still compatible with the standardized interface for \ac{NN} verification tools~\cite{VNNLIB}.
We call this partitioning of the input space a \emph{mosaic} and,
in reminiscence of this analogy, we call the individual queries \emph{azulejos} (\textipa{[A.Tu'le.xo]}, see e.g. regions in shades of gray in \Cref{fig:ncubev}).
$\normalizeAlg{}$ guarantees that reachability-based \ac{OLNNV} tools do not explore the same input region multiple times (see \emph{flatness} result in \Cref{lem:decomp:minimum}; \Cref{apx:mosaic}).
For our ACAS case study, naive rewriting (as done by e.g. DNNV~\cite{Shriver2021}) may produce up to 39 trillion propositionally feasible queries.
In contrast, \normalizeAlg{} only produces 19k queries (see \Cref{tab:dnnv_comp} in \Cref{subsec:eval:closed_loop}).
\normalizeAlg{} also separates nonlinear constraints $q_n$ (must be checked outside \ac{OLNNV}) from linear constraints $q_l$ (can be passed to \ac{OLNNV}).

On the technical side, Mosaic proceeds by executing a SAT solving based DPLL(T) loop until a satisfiable conjunction of \emph{input} constraints is found. At this stage, we fix the conjunction's linear input constraints (i.e. the azulejo) and an inner loop enumerates conjunctions over mixed/output constraints that are satisfiable in combination with the fixed azulejo.
For each such conjunction, we save the conjunction of linear mixed/output constraints.
This results in a linear, normalized Open-Loop NNV query (conjunction over input, disjunctive normal from over output).
We employ a similar inner loop to enumerate satisfiable conjunctions of nonlinear constraints to later check counterexamples via SMT solving (see Retaining completeness).
At each step, we interleave propositional and theory solving to discard conjunctions unsatisfiable in real arithmetic as early as possible.

\begin{algorithm}[t]
    \caption{Verification of nonlinear queries on piece-wise linear \acp{NN}:
    \enumerateAlg{} internally uses off-the-shelf \ac{OLNNV} tools for the verification of linear normalized queries on \acp{NN}.}
    \label{algorithm:nnnv:the_algorithm}
    \begin{algorithmic}
    \Require{Formula $p \in \formulaset{\reals}$, \text{ranges} $R$, piece-wise linear \ac{NN} $g$}
    \State $p_o \leftarrow \overapproxAlg\mleft(p,R\mright)$
    \Comment{Generate linearized $p$}
    \ForAll{$\left(q_l,q_n\right) \in \normalizeAlg\mleft(p_o\mright)$}
        \Comment{Iterate over normalized queries}
        \ForAll{$\left(\iota,\omega\right) \in \enumerateAlg\mleft(g, q_l\mright)$}
        \Comment{uses \generalizeAlg{} \& \ac{OLNNV}}
            \If{\Call{Filter}{$q_l \land q_n,\iota,\omega$}=\texttt{concrete}}
                \Return \texttt{unsafe}
            \EndIf
        \EndFor
    \EndFor
    \State \Return \texttt{safe}
    \Comment{No concrete counterexamples found}
    \end{algorithmic}
\end{algorithm}

\paragraph[Linearize]{\overapproxAlg{}.}
\looseness=-1
To verify polynomial specifications via off-the-shelf (linear) \ac{OLNNV}, we need to soundly approximate the query's nonlinear constraints.
In principle, we could perform this approximation for each azulejo (see \normalizeAlg{}) separately.
To this end, consider an atomic polynomial constraint $p\left(x\right) \leq 0$ which is part of a query.
If there are two azulejos with the constraints $p\left(x\right) \leq 0$ and $p\left(x\right) > 0$,
separate approximation would lead to the constraint $\overline{p}\left(x\right) \leq 0$ and $\underline{p}\left(x\right) > 0$ (with $\overline{p}$/$\underline{p}$ resp. over/under-approximations of $p$).
This would \emph{duplicate} the exploration of the area in between $\overline{p}$ and $\underline{p}$ and can lead to an exponential blowup for many approximations.
Instead, we use a \emph{global} piece-wise approximation via OVERT~\cite{Sidrane2021} (orange lines around the blue nonlinear constraint on \Cref{fig:ncubev}) and integrate the approximate constraints into the original query via implications (e.g. we add $\overline{p}\left(x\right) \leq 0 \implies p\left(x\right)\leq 0$).
These additional linear constraints are then automatically enumerated via \normalizeAlg{} (see azulejos for $\mathrm{case}_1$ in \Cref{fig:ncubev}).
The linearization happens for input as well as output constraints.
When passing the \ac{OLNNV} query to the off-the-shelf tool, we soundly omit the nonlinear constraints, thus only leaving behind linear constraints $q_l$ (green and orange in \Cref{fig:ncubev}).

\paragraph{Retaining completeness.}
Without further efforts, \normalizeAlg{} and \overapproxAlg{} together yield a sound algorithm to check nonlinear, not normalized \ac{OLNNV} queries, but not a complete one:
A reachability analysis via \ac{OLNNV} for a given azulejo may produce spurious counterexamples that are an artifact of the linearization (see \Cref{fig:ncubev})
The key insight to achieve completeness is the observation that for piece-wise \emph{linear} \acp{NN} any concrete counterexample generated via \ac{OLNNV} corresponds to a \emph{counterexample region} $\iota$ (yellow polytope on the left of \Cref{fig:ncubev}) on which the \ac{NN} reduces to an affine transformation $\omega$ (obtained by fixing piece-wise functions to the linear segment of the concrete counterexample, i.e. we fix the value of the $\mathds{1}$ functions; see also \Cref{subsec:background:nnv}).
This insight can be used in two ways:
First, we can use it to enumerate all counterexample regions $\iota$ (by adjusting the tool's internal enumeration and/or via the algorithm \enumerateAlg{}; see \Cref{subsec:enumerate}).
Secondly, for a given counterexample region, we can check whether there exists a concrete counterexample to the nonlinear specification via SMT solving (in \Cref{algorithm:nnnv:the_algorithm} this is performed by \filterAlg{}; see \Cref{subsec:enumerate}).
By exploiting the affine transformation $\omega$, our SMT encoding for nonlinear constraints has at most $I$ free variables (for $I$ input dimensions of the \ac{NN}) and is therefore \emph{significantly more tractable} than an encoding of the entire \ac{NN} in SMT.

Using the components outlined above (for details see \Cref{apx:mosaic}), we prove the soundness and completeness of \chFourFancyName{}.
Completeness turns out to also be of \emph{practical relevance} as approximation alone would have failed to verify the DNC \ac{NN} of the ACAS benchmark discussed in \Cref{sec:evaluation}.
\begin{theoremE}[Soundness and Completeness][end,text link=]
\label{thm:nnnv:sound_and_complete}
Let $g$ be a piece-wise linear \ac{NN}, $p$ be a real arithmetic formula and $R$ variable ranges for all in- and output variables of $g$.
\Cref{algorithm:nnnv:the_algorithm} returns \texttt{unsafe} iff there exists an input $z \in \reals^I$ such that $\left(z,g\left(z\right)\right)$ is in the range $R$ and $p\mleft(z,g\mleft(z\mright)\mright)$ is satisfied.
\end{theoremE}
\begin{proofE}
\label{proof:nnnv:sound_and_complete}
This proof assumes the results from \Cref{subsec:query_decompose,subsec:enumerate,subsec:overapprox}.
We begin by proving soundness, i.e. if \Cref{algorithm:nnnv:the_algorithm} returns safe, then there exists no counterexample.
Consider a counterexample region found by \enumerateAlg{}.
\Cref{lem:counter-example-filtering} tells us that this counterexample can only be concrete if formula $\nu$ is satisfied.
This is the check performed by \filterAlg{}.
Thus, \Cref{algorithm:nnnv:the_algorithm} only skips a counterexample region if it is not concrete.
Further, we know through \Cref{def:enumerateAlg} that all counterexamples are returned by the procedure for a given query $q_l$.
Further, we know that the disjunction over all $q_l\land q_n$ returned by \normalizeAlg{} is equivalent to its input $p_o$ (\Cref{lem:decomp:correct}) and thus the disjunction over all $q_l$ is an over-approximation thereof.
Finally, \overapproxAlg{} returns a formula $p_o$ which is equivalent to the input query $p$ (\Cref{lem:linearization:equivalence}).
Therefore, any counterexample of $p$ must also be a counterexample of some $q_l$ returned by \normalizeAlg{}.
Consequently, we iterate over all possible counterexamples and only discard them if they are spurious.
Thus, our algorithm is sound.

We now turn to the question of completeness, i.e. we prove that any time \Cref{algorithm:nnnv:the_algorithm} returns \texttt{unsafe} then there is indeed a concrete counterexample of $p$.
First, remember that \Cref{lem:linearization:equivalence} ensures that $p$ and $p_o$ are equivalent.
Furthermore, \Cref{lem:decomp:correct} ensures that the disjunction over all $q_l \land q_n$ generated by \normalizeAlg{} is equivalent to $p_o$.
Assume we found a counterexample.
The algorithm will return \texttt{unsafe} iff \filterAlg{} returns that the counterexample is concrete.
According to \Cref{lem:counter-example-filtering} we know that this is only the case if there is indeed a concrete counterexample for $q_l \land q_n$.
Since this counterexample then also satisfies $p_o$ (see above), we only return \texttt{unsafe} if \filterAlg{} found a concrete counterexample for $p$.
As real arithmetic is decidable and all other procedures in the algorithm terminate as well, this yields a terminating, sound and complete algorithm.
\end{proofE}

In \Cref{sec:evaluation} we build upon nnenum~\cite{bak2020improved,bakOverapprox} which can enumerate all counterexample regions of an \ac{NN} with $N$ ReLU activations in time $\mathcal{O}\big(2^N\big)$, and upon  cylindrical algebraic decomposition (CAD~ \cite{DBLP:journals/cca/Collins74}) with complexity $\mathcal{O}\big(2^{2^V}\big)$ for $V$ variables\footnote{For simplicity of analysis we fix the complexity's other variables (e.g. maximal degree).}.
Assuming $M$ atomic formulas in the \ac{OLNNV} query and $I$ input dimensions this yields a worst-case runtime of $\mathcal{O}\big(2^{M+\mathbf{N}+{2^I}}\big)$.
This is an exponential improvement over naive $\mathcal{O}\big(2^{M+2^{\mathbf{N}+I}}\big)$ CAD encodings as $N >> I$.
In practice the performance is even better; usually, \normalizeAlg{} explores fewer queries, nnenum returns less counterexample regions and SMT solving tends to perform very well for the small input dimensions $I$ of \ac{NNCS} control.

\section{Evaluation}
\label{sec:evaluation}
\sepfootnotecontent{SNNT}{%
See \url{https://github.com/samysweb/NCubeV} or \cite{samuel_teuber_2024_13922169}
}
\looseness=-1 %
We implemented Mosaic for ReLU \acp{NN} in a new Julia~\cite{Bezanson_Julia_A_fresh_2017} tool called N\sepfootnote{SNNT}V
based on the software packages nnenum~\cite{bak2020improved,bakOverapprox}, PicoSAT~\cite{PicoSAT,PicoSatJL} and Z3~\cite{Z3overall,Z3nonlinear}.
We provide wall-clock times on an AMD Ryzen 7 PRO 5850U CPU (\ac{SNNT} is sequential; nnenum uses multithreading).
The evaluation presented in this section focuses on vertical airborne collision avoidance (VCAS), with other experiments in \Cref{apx:evaluation}.
Airborne Collision Avoidance Systems try to recognize plane trajectories that might lead to a \ac{NMAC} with other aircrafts and advise the pilot to avoid such collisions.
\acp{NMAC} are defined as two planes (ownship and intruder) flying closer than 500 ft horizontally or 100 ft vertically.
Currently, the \ac{FAA} develops a new airborne collision avoidance system called \acf{ACASX}.
Prior work by Jeannin et al.~\cite{Jeannin2017} showed a nondeterministic, provably safe \ac{dL} envelope for airborne collision avoidance.
While the original proposal for ACAS X~\cite{holland2013optimizing,kochenderfer2012next} was shown to be unsafe~\cite{Jeannin2017}, 
the correctness of a VCAS \ac{NN} implementation~\cite{julian2016policy,julian2020safe} was proven~\cite{JulianKDVCASSafe} and disproven~\cite{Julian2019a} in special cases.
The proposed \acp{NN} contain 6 hidden layers with 45 neurons each and produce one of 9 collision avoidance advisories (\emph{Strengthen Climb to at least 2500 ft/min} (SCL2500) to \emph{Strengthen Descent to at least 2500ft/min} (SDES2500); see \Cref{tab:acas_table} in \Cref{apx:acas_table} for a list of allowed advisories).
The objective of the \acp{NN} is to ensure safety while minimizing the number of alerts sent to the pilot.
We present an \emph{exhaustive} analysis of the VCAS \ac{NNCS} for level flight intruders.

\begin{wraptable}[15]{r}{0.6\textwidth}
    \centering
    \caption{Verification of ACAS \acp{NN} for level flight: Previous advisory (=Prev. Adv.), runtime; number of counterexample (=CE) regions and time to the discovery of the first CE.}

    \begin{tabular}{l|l|l|l|l}
        Prev. Adv. & Status & Time & CE regions & First CE\\\hline
        DNC & \textbf{safe} & 0.35 h & --- & --- \\
        DND & \textbf{safe} & 0.28 h & --- & --- \\
        DES1500 & \textbf{unsafe} & 5.45 h & 49,428 & 0.04 h\\
        CL1500 & \textbf{unsafe} & 5.18 h & 34,658 & 0.08 h\\
        SDES1500 & \textbf{unsafe} & 4.05 h & 5,360 & 0.97 h\\
        SCL1500 & \textbf{unsafe} & 4.89 h & 11,323 & 0.36 h\\
        SDES2500 & \textbf{unsafe} & 3.66 h & 5,259 & 1.39 h\\
        SCL2500 & \textbf{unsafe} & 4.45 h & 7,846 & 0.53 h\\
    \end{tabular}
    \label{tab:acas}
\end{wraptable}%
We provide additional experimental results in \Cref{apx:evaluation}.
First, we demonstrate the feasibility of our approach for the running example of \ac{ACC}.
Depending on NN size and chosen linearization, our approach can verify or exhaustively enumerate counterexamples for the NNCS in 47 to 300 seconds.
For a case study on Zeppelin steering under (uniformly sampled) wind perturbations, we adapt a differential hybrid games formalization~\cite{Platzer2017} to analyze an NN controller trained by us.
Here, we encountered a controller that showed very positive empirical performance while being provably unsafe in large parts of the input space:
While performing very well on average, the control policy was vulnerable to unlikely wind perturbations -- an issue we only found through our verification.
For \ac{ACC}, we also perform a comparison to other techniques:
While Closed-Loop techniques are useful for the analysis of bounded-time safety, their efficiency greatly depends on the system's dynamics and the considered input space.
Our infinite-time horizon approach can be more efficient than Closed-Loop techniques as it evades the necessity to analyze the system's dynamics along with the NN (see \Cref{tab:clnnv_comp}).
Usually, it is desirable to show infinite-time safety on the entire (controllable) state space.
However, the approximation errors incurred via prior closed-loop NNV techniques prohibit this as they will either ignore states inside the controllable region or allow unsafe actions pushing the system outside its controllable region.
Conversely, SMT-based techniques do not have these approximation issues, but cannot scale to NNs of the size analyzed in this work.
We also provide a conceptual comparison demonstrating the efficiency of the Mosaic procedure for normalized query generation over DNNV's expansion-based algorithm (see \Cref{tab:dnnv_comp}), naive SMT solving (\Cref{tab:smt_comp}) and Genin et al.'s tailored ACAS approach~\cite{Genin2022}.

\paragraph{Verification Results for VCAS.}
\looseness=-1
We use the nondeterministic control envelope and loop invariant by Jeannin et al.~\cite[Thm. 1]{Jeannin2017} to analyze safety for intruders in level flight (i.e. intruder vertical velocity is 0).
For the same reasons as prior work~\cite{Jeannin2017}, we ignore Clear-of-Conflict advisories.
The \ac{OLNNV} queries obtained via \chThreeFancyName{} had up to 112 distinct atoms and trees up to depth 9.
To determine the maximal output neuron, \ac{OLNNV} queries initially contain atoms sorting the \ac{NN}'s outputs $x_1,\dots,x_n$ (e.g., $x_1\leq x_2$).
To avoid enumerating all permutations, %
we perform symmetry elimination via an atomic predicate encoding that some output $i$ is maximal.
\looseness=-1
We analyzed the \emph{full range} of possible \ac{NN} inputs for intruders in level flight (relative height $\left|h\right|\leq 8000 \text{ft}$, ownship velocity $\left|v\right| \leq 100 \text{ft/s}$ and time to \ac{NMAC} $6\text{s} \leq \tau \leq 40\text{s}$).
Our results are in \Cref{tab:acas}:
\textbf{safe} implies that the \ac{NN}'s advisories (other than Clear-of-Conflict) in this scenario \emph{never} lead to a collision when starting within the invariant.
The safety for DNC was only verifiable through SMT filtering (approximation yielded spurious counterexamples).
This \emph{underscores the importance of \chFourFancyName{}'s completeness} and implies that \cite{Julian2019a} is insufficient to prove safety.
A non-exhaustive analysis for non-level flight yielded counterexamples even for DNC/DND (see \Cref{apx:acas}).
For unsafe level flight scenarios, we exhaustively characterize unsafe regions.
This characterization goes far beyond characterizations in prior work~\cite{Julian2019a} that were generated using manual approximation and resulted in (non-exhaustive) point-wise characterizations.
\Cref{fig:acas} shows a concrete avoidable \ac{NMAC} (more examples are in \Cref{apx:acas}).

\begin{figure}[t]
    \centering
    \includegraphics[width=0.7\textwidth]{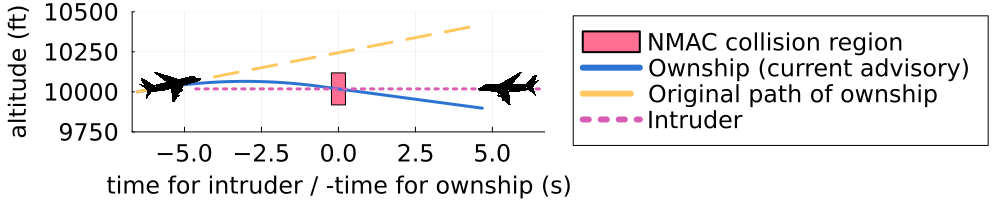}
    \caption{An unsafe advisory by the Airborne Collision Avoidance \ac{NN}: After a previous advisory to climb at least 1500ft/min, the \ac{NN} advises to reverse vertical direction (``Strengthen Descent to at least 1,500 ft/min''). This leads to an NMAC 6 seconds later. More examples are in \Cref{apx:acas}.}
    \label{fig:acas}
\end{figure}

\paragraph{Scalability.}
\looseness=-1
\ac{SNNT} provides guarantees for an \ac{NN}'s full input space.
Hence, \ac{SNNT} is not directly comparable to, e.g., the scalability of local robustness verifiers, that, while sometimes scaling to hundred thousands of \relu{}s~\cite{DBLP:journals/corr/abs-2312-16760}, only analyze tiny fractions of the input space~\cite[Sec. 6]{tran2020verification}.
Contemporary work on global properties outside \acp{NNCS} has been scaled to 100 \relu{} nodes~\cite{DBLP:conf/cav/AthavaleBCMNW24}.
Related work on infinite-time \acp{NNCS} (see also \Cref{sec:related}) has not been scaled beyond 30 nodes in similar settings~\cite{DBLP:conf/nips/LechnerZCH21,DBLP:conf/nips/ChangRG19}.
The largest \acp{NN} verified by us so far had 270 \relu{}s, indicating \ac{SNNT} is at the frontier of State-of-the-Art scalability for \emph{global} properties.
Some \ac{NNCS} applications (e.g. ACAS~\cite{julian2016policy}) turn to \acp{NN} to efficiently encode complex strategies in mid-scale \acp{NN}.
\ac{SNNT} scales to \acp{NN} of this kind.

\section{Related Work}
\label{sec:related}
\paragraph{Shielding.}
\looseness=-1
Justified Speculative Control~\cite{Fulton2018} is closely related in its use of \ac{dL}.
However, we verify the \ac{NNCS} \emph{a priori} instead of treating ML models as a black box and \emph{a posteriori} using runtime enforcement techniques~\cite{Mitsch2016,DBLP:conf/birthday/KonighoferBEP22}.
Shielding can also be used to train models such that they (probabilistically) conform to a shield/monitor~\cite{DBLP:conf/aaai/AlshiekhBEKNT18}.
Training methodologies are beyond the scope of this paper.

\paragraph{Barrier Certificates.}
An orthogonal direction of research explores learning Neural Barrier/Lyapunov Functions to prove safety properties~\cite{DBLP:journals/trob/DawsonGF23}.
Although initially used for ``pure'' dynamical and hybrid systems (without NNs)~\cite{noroozi2008generation,DBLP:conf/corl/RichardsB018,DBLP:conf/tacas/PeruffoAA21,DBLP:journals/csysl/AbateAGP21,DBLP:conf/tacas/AhmedPA20,DBLP:conf/hybrid/AbateAEGP21,DBLP:conf/nips/Zhang0V023,DBLP:conf/hybrid/ZhaoC0SYLTCL21},
the methods have since been extended to NNCS with discrete~\cite{DBLP:conf/ijcai/BacciG021,DBLP:conf/nips/LechnerZCH21,DBLP:conf/ciss/ChenFMPP21}
and continuous~\cite{DBLP:conf/nips/ChangRG19,DBLP:conf/corl/DawsonQGF21,DBLP:conf/iclr/QinZCCF21} time behavior.
While the former works can only approximate continuous time behavior, the latter techniques use off-the-shelf SMT solvers (see also \Cref{subsec:eval:closed_loop}) for certificate verification which severely limits scalability.
While some works ignore verification entirely~\cite{DBLP:conf/corl/DawsonQGF21,DBLP:conf/iclr/QinZCCF21}, the remaining work only considered \emph{linear single-layer} \acp{NN}~\cite[Appendix]{DBLP:conf/nips/ChangRG19} verified with dReal (see SMT comparison in \Cref{subsec:eval:closed_loop}).
Dawson et al.~\cite{DBLP:conf/corl/DawsonQGF21} note ``scalable verification for learned certificate functions remains an open problem''.
Using \ac{SNNT} as an alternative to SMT for certificate verification is future work.
\chThreeFancyName{} evades the necessity to \emph{learn} barrier certificates by reusing established control-theory literature.

\paragraph{Open-Loop NNV.}
\looseness=-1
VEHICLE~\cite{Daggitt2022} integrates \ac{OLNNV} with Agda. %
However, VEHICLE only allows importing normalized, linear properties which limits applicability to \ac{CPS} verification in realistic settings.
\Ac{OLNNV} tools \cite{zhang2018efficient,xu2020automatic,xu2020fast,wang2021beta,Henriksen2020,Henriksen2021,katz2017reluplex,katz2019marabou,bunel2020branch,bak2020improved,DBLP:conf/cav/ElboherGK20}
do not consider the physical environment and thus \emph{cannot} guarantee the safety of an \ac{NNCS}. 
Even in cases where such methods allow nonlinear behavior in activation functions, they do not admit the verification of arbitrary polynomial constraints over the input and output space.
Most methodologies could be integrated into \chFourFancyName{}'s framework, i.e. we can lift complete off-the-shelf \ac{OLNNV} tools to verify polynomial constraints with arbitrary structure.
DNNV~\cite{Shriver2021} proposed an approach for \ac{OLNNV} query normalization using a simple expansion algorithm.
DNNV has the same limitations as all \ac{OLNNV} tools (no \ac{NNCS} analysis; no nonlinear constraints)
and is less efficient than \chFourFancyName{} w.r.t. \ac{NN} reachability analysis (see \Cref{sec:nnnv}).
Pre-image computation~\cite{kotha2024provably,10.1007/978-3-031-57256-2_1,matoba2020exact} computes input regions producing a fixed \ac{NN} prediction.
For efficiency, our work constrains the input space with invariants and value ranges---this efficiency would be lost by a backward computation alone.

\paragraph{Related Techniques.}
\looseness=-1
Unlike \cite{qin2018verification,DBLP:conf/nips/BalunovicBSGV19}, we support \emph{arbitrary} polynomial constraints and retain completeness.
Moreover, these works do not support arbitrary logical structure and represent \ac{OLNNV} techniques unable to analyze \acp{NNCS}.
Some prior work used techniques for constructing counterexample regions~\cite{DBLP:conf/nips/ZhangSS18,DBLP:conf/iclr/Dimitrov0GV22} but only for \emph{individual} datapoints---neither to compute exhaustive characterizations nor to regain completeness for incomplete verifiers.
Unlike classical DPLL(T)~\cite{DPLLT}, $\normalizeAlg{}$ is tailored to theory-solving w.r.t. reachability analyzers. To this end, $\normalizeAlg{}$ groups output regions with the same input constraints which deduplicates work (see \Cref{sec:nnnv}).

\paragraph{Closed-Loop NNV.}
\Ac{CLNNV} tools~\cite{Forets2019JuliaReachReachability,Schilling2021,tran2019star,tran2020neural,9093970,IvanovACM21,Ivanov2021,Huang2019,Fan2020,Dutta2019,Sidrane2021,Akintunde2022} only consider a fixed time horizon and thus cannot guarantee infinite-time horizon safety (see \Cref{subsec:eval:closed_loop}).
Unlike \cite{Julian2019a,Genin2022,Bak2022}, our approach is automated and applicable to \emph{any} \ac{CPS} expressible in \ac{dL} (not just ACAS X; see case studies in \Cref{apx:evaluation}).
Other approaches verify simplified control outputs~\cite{Genin2022}; rely on hand-crafted approximations while lacking exhaustive counterexamples characterizations~\cite{Julian2019a}; or require quantization effectively analyzing only a surrogate system instead~\cite{Bak2022}.

\section{Conclusion and Future Work}
\label{sec:conclusion}
\looseness=-1
This work presents \chThreeFancyName{}, the first technique exploiting \ac{dL} contracts to prove safety of \acp{NNCS} with piece-wise Noetherian \acp{NN}.
\chThreeFancyName{} requires \ac{OLNNV} tools capable of verifying non-normalized polynomial properties that did not exist. %
Thus, with \chFourFancyName{} we present an efficient, \emph{sound and complete} approach for the verification of such properties on piece-wise linear \acp{NN}.
We implemented \chFourFancyName{} for ReLU \acp{NN} in the tool \ac{SNNT} and demonstrate
the applicability and scalability
of our approach on multiple case studies (\Cref{sec:evaluation} and \Cref{apx:evaluation}).
The application to \acp{NNCS} by Julian et al.~\cite{julian2016policy,julian2020safe} shows that our approach scales even to intricate, high-stakes applications such as airborne collision avoidance.
Our results underscore the categorical difference of our approach to \ac{CLNNV} techniques.
Overall, we demonstrate an efficient and generally applicable approach
that opens the door for developing of goal-oriented \emph{and} infinite-time horizon safe \acp{NNCS} in the real world.

\paragraph{Future Work.}
We believe there is potential for engineering and algorithmic improvements that could further improve the performance of \ac{SNNT}.
Our implementation could also be extended to other piece-wise linear activation functions.
We would also like to explore the efficiency of \ac{SNNT} for non-polynomial specifications w.r.t. suitable SMT solvers.
\chFourFancyName{} may be of interest even beyond NNCS verification in \chThreeFancyName{}.
Since, e.g., neural barrier certificate verification also requires nonlinear \ac{OLNNV}, \chFourFancyName{} could be equally applicable in this context.
Finally, it would be interesting to apply our approach to further case studies.
To this end, the bottle-neck is currently the limited availability of \ac{NNCS} which are safe w.r.t. an infinite-time horizon.

\section*{Acknowledgements}
This work was supported by funding from the pilot program Core-Informatics of the Helmholtz Association (HGF) and by an Alexander von Humboldt Professorship.
Part of this research was carried out while Samuel Teuber was funded by a scholarship from the International Center for Advanced Communication Technologies (interACT) in cooperation with the Baden Württemberg Foundation.

\bibliography{neurips_2024}

\newif\ifchecklist
\checklisttrue

\clearpage
\appendix

\section{Additional Background}
\label{apx:additional_background}
This section provides additional background on Differential Dynamic Logic and the status of various activation functions w.r.t. our classification in \Cref{tab:nn_classes}.
\subsection{Differential Dynamic Logic}
\label{apx:subsec:background:dl}
\Acf{dL}~\cite{Platzer2017a,Platzer2020,DBLP:conf/lics/Platzer12a,DBLP:journals/jar/Platzer08} can analyze models of hybrid systems that are described through \emph{hybrid programs}.
The syntax of hybrid programs with Noetherian functions is defined by the following grammar, where the term $e$ and formula $Q$ are over real arithmetic with Noetherian functions:
\begin{equation}
\label{apx:eq:hp_grammar}
\grammar{\alpha,\beta}{\left. x \coloneqq e ~\middle\mid~ x \coloneqq * ~\middle\mid~ ?Q ~\middle\mid~ x'=f(x)\&Q ~\middle\mid~ \alpha \cup \beta ~\middle\mid~ \alpha;\beta ~\middle\mid~ \alpha^*\right.}
\end{equation}
\looseness=-1
The semantics of hybrid programs are defined by a transition relation over states in the set $\dLstatespace$, each assigning real values to all variables.
For example, the assignment state transition relation is defined as 
$\dLinterpretationapp{x \coloneqq e} = \{ \left(\nu,\omega\right) \in \dLstatespace^2 \mid \omega = \interpretationsubstitute{\nu}{x}{\nu\mleft(e\mright)} \} $
where $\interpretationsubstitute{\nu}{x}{\nu\mleft(e\mright)}$ denotes the state that is equal to $\nu$ everywhere except for the value of $x$, which is modified to $\nu\mleft(e\mright)$.
The other programs in the same order as in~(\ref{apx:eq:hp_grammar}) describe nondeterministic assignment of $x$, test of a predicate $Q$, continuous evolution along the differential equation within domain $Q$, nondeterministic choice, sequential composition, and nondeterministic repetition.
For a given program $\alpha$,
we distinguish between bound variables $\boundedvars{\alpha}$ and free variables $\freevars{\alpha}$ where bound variables are (potentially) written to, and free variables are read.
The formula $\left[\alpha\right] \psi$ expresses that $\psi$ is always satisfied after the execution of $\alpha$ and $\left\langle \alpha\right\rangle \psi$ that there exists a state satisfying $\psi$ after the execution of $\alpha$.
If a state $\nu$ satisfies $\psi$ we denote this as $\nu \vDash \psi$.
\ac{dL} comes with a sound and relatively complete proof calculus~\cite{DBLP:journals/jar/Platzer08,Platzer2017a,Platzer2020} as well as the interactive theorem prover \acf{keymaerax}~\cite{Fulton2015}.

\paragraph{ModelPlex.}
As demonstrated in the example, many safety properties for \acp{CPS} can be formulated through a \ac{dL} formula with a loop in which $\alpha_{\text{ctrl}}$ describes the (discrete) software, and $\alpha_{\text{plant}}$ describes the (continuous) physical environment:
\begin{equation}
\label[formula]{apx:eq:proven_safety_contract}
\phi \implies \left[ \left( \alpha_{\text{ctrl}}; \alpha_{\text{plant}} \right)^* \right] \psi .
\end{equation}
\looseness=-1
$\phi$ describes initial conditions, and $\psi$ describes the safety criterion 
to be guaranteed when following the control-plant loop.
To ensure that the behavior of controllers and plants in practice match all assumptions represented in the contract, \ac{modelplex} shielding~\cite{Mitsch2016} synthesizes correct-by-construction monitors for \acp{CPS}.
\Ac{modelplex} can also synthesize a correct controller monitor formula $\controllerFormula$ (\Cref{def:correct_monitor}).
The formula $\controllerFormula$ encodes a relation between two states.
If $\controllerFormula$ is satisfied, then the variable change from $x_i$ to $x_i^+$ corresponds to behavior modeled by $\alpha_{\text{ctrl}}$, i.e. the change upholds the guarantee from \Cref{apx:eq:proven_safety_contract}.
To reason about this state relation, we say that a state tuple $\left(\nu,\omega\right)$ satisfies a formula $\zeta$ (denoted as $\left(\nu,\omega\right) \vDash \zeta$) iff
$\interpretationsubstitute{\nu}{
\hphantom{\omega}x_1^+\ldots x_n^+
}{
\omega\mleft(x_1\mright)\ldots\omega\mleft(x_n\mright)
} \vDash \zeta$ (i.e. $\nu$ with the new state's value $\omega\mleft(x_i\mright)$ as the value of $x_i^+$ for all $i$, satisfies $\zeta$).
\begin{definition}[Correct Controller Monitor~\cite{Mitsch2016}]
\label{def:correct_monitor}
A controller monitor formula $\controllerFormula$ with free variables $x_1,x_1^+,\dots,x_n,x_n^+$ is called \emph{correct} for the hybrid program controller $\alpha_{\text{ctrl}}$ with bound variables $x_1,\dots,x_n$ iff the following \ac{dL} formula is valid:
$
\controllerFormula \implies \langle \alpha_{\text{ctrl}} \rangle \bigwedge_{i=1}^n x_i = x_i^+
$.
\end{definition}

\subsection{Classification of Activation Functions}
\label{subsec:activation_table}
\Cref{tab:activation_functions} is meant to provide a brief overview on the status of various common activation functions w.r.t. the classes from \Cref{tab:nn_classes}.
The chosen activation functions are representative examples from the comprehensive survey by Kunc et al.~\cite{kunc2024decades}.
For piece-wise polynomial and Noetherian functions, any activation function represented by a polynomial over the activation functions presented in \Cref{tab:activation_functions} is resp. piece-wise polynomial or Noetherian.
For piece-wise linear functions, any linear combination is resp. piece-wise linear.
This table does not claim to be comprehensive---many more activation functions can be classified into one of the categories.

\begin{table}[t]
    \centering
    \begin{tabularx}{\textwidth}{l | X}
         Activation Function & Justification\\\hline\hline
         \multicolumn{2}{l}{\textbf{ReLU \acp{NN}}}\\\hline
         ReLU & ReLU \\\hline
         \multicolumn{2}{l}{\textbf{piece-wise linear \acp{NN}}}\\\hline
         LeakyReLU~\cite[3.6.2]{kunc2024decades} & Split into $<0$ and $\geq 0$.\\
         HardTanh~\cite[3.6.18]{kunc2024decades} & Split into 3 regions.\\
         \hline
         \multicolumn{2}{l}{\textbf{piece-wise polynomial \acp{NN}}}\\\hline
         Square-based activation functions~\cite[3.8]{kunc2024decades} &
         by definition. \\
         Polynomial universal activation function~\cite[3.16]{kunc2024decades} &
         by definition for integer exponents.\\
         \hline
         \multicolumn{2}{l}{\textbf{piece-wise Noetherian \acp{NN}}}\\\hline
         Sigmoid~\cite[3.2]{kunc2024decades} & $\sigma'\mleft(x\mright)=\sigma\mleft(x\mright)\left(1-\sigma\mleft(x\mright)\right)$\\
         Tanh~\cite[3.2]{kunc2024decades} & $\mathrm{tanh}\left(x\right) = 2\sigma\left(2x\right)-1$\\
         Arctan~\cite[3.2.4]{kunc2024decades} & $\mathrm{arctan}'\left(x\right)=1/\left(1+x^2\right)$\\
         SiLU~\cite[3.3]{kunc2024decades} & Polynomial over $\sigma$\\
         $\exp\left(x\right)$ & $\exp'\left(x\right)=\exp\left(x\right)$\\
         GELU~\cite[3.3.1]{kunc2024decades} & In a Noetherian Chain with $\exp$ and the gaussian error funciton\\
         Approximate GELU~\cite[3.3.1]{kunc2024decades} & Polynomial over $\sigma$ or $\mathrm{tanh}$\\
         Softmax~\cite[3.5]{kunc2024decades}\footnote{The analysis of softmax can often be avoided by using the observation that its application does not change the order of outcomes, i.e. for classification-like tasks we can use the maximum before Softmax application instead.} & Polynomial over $\exp$.
    \end{tabularx}
    \caption{Overview on classification of activation functions.}
    \label{tab:activation_functions}
\end{table}

\section{Mosaic}
\label{apx:mosaic}
\looseness=-1
An overview of the algorithm is given in \Cref{fig:ncubev}:
A piece-wise linear \ac{NN} $g$ is a function which maps from an input space (left) to an output space (right).
We consider a part of the input space that is constrained by linear (orange) and nonlinear constraints (blue).
As our query is not normalized, it may talk about multiple parts of the input space, e.g. in our case the two sets labeled with $\text{case}_1$ and $\text{case}_2$.
For any such part of the input space, say $\text{case}_1$, we have a specification about \emph{unsafe} parts of the output space which must not be entered (red dashed areas on the right).
For classical \ac{OLNNV} the task is then, given a single input polytope, to compute the set of reachable outputs for $g$ and to check whether there exists an output reaching an unsafe output polytope.
In our case the task is more complicated, because the input is not a polytope, but an arbitrary polynomial constraint.
Moreover, for each polynomial input constraint ($\text{case}_1$ and $\text{case}_2$ in \Cref{fig:ncubev}) we may have different nonlinear unsafe output sets.
In order to retain soundness, we over- \emph{and} underapproximate nonlinear constraints (see turquoise linear approximations around the blue and red nonlinear constraints in \Cref{fig:ncubev}).
Once all nonlinear queries have approximations, we generate a \emph{mosaic} of the input space where each \emph{azulejo} (i.e. each input region) has its own normalized \ac{OLNNV}-query (polytope over the input; disjunction of polytopes over the output).
We must not only split between the two original cases ($\text{case}_1$ and $\text{case}_2$), but also between different segments of the approximating constraints (see the polytopes on the left in four shades of gray).
Each normalized query has associated nonlinear constraints that must be satisfied, but cannot be checked via off-the-shelf \ac{OLNNV} tools.
Using the normalized, linear \ac{OLNNV} queries we then instrument off-the-shelf tools to check whether any overapproximated unsafe region (turquoise on the right) is reachable.
The amber colored regions represent parts of the outputs reachable by $g$:
As can be seen by the two red dots, a reachable point within the overapproximated unsafe region may be a concrete unsafe output, or it may be spuriously unsafe due to the overapproximation.
To retain completeness, we need to exhaustively filter spurious counterexamples.
This is achieved by \emph{generalizing} the counterexample to a  region around the point in which the behavior of $g$ is equivalent to a single affine transformation.
Such counterexample regions always exist 
due to $g$'s piece-wise linearity.
For example, in \Cref{fig:ncubev} $g$ affinely maps the input space triangle (left) to the upper output space triangle (in amber on the right).
We can then check for concrete counterexamples in this region w.r.t. the affine transformation using an SMT solver.
This avoids the need to encode the entire \ac{NN} in an SMT formula.
By excluding explored regions, we can then enumerate all counterexample regions and thus characterize the unsafe input set.

\looseness=-1
The approach is outlined in \Cref{algorithm:nnnv:the_algorithm} and proceeds in four steps, all of which will be presented in detail throughout the sections below:
First, \overapproxAlg{} generates approximate linearized versions for all nonlinear atoms of $p$ on a bounded domain and enriches the formula with these constraints ($p_o$). %
Next, \normalizeAlg{} generates a mosaic of $p_o$'s input space
where each azulejo (i.e. each input space region) has an associated linear normalized query $q_l$.
Each $q_l$ is paired with an associated disjunctive normal form of nonlinear constraints $q_n$.
The disjunction over all $q_l \land q_n$ is equivalent to the input query $p_o$ and the disjunction over all $q_l$ overapproximates \overapproxAlg{}'s input $p$.
Each of the linear queries $q_l$ is processed by \enumerateAlg{} which internally uses an off-the-shelf \ac{OLNNV} tool to enumerate all counterexample regions for a given query.
Each counterexample region is defined through a polytope in the input space $\iota \subset \mathbb{R}^I$ and an affine mapping to the output space $\omega:\mathbb{R}^I \to \mathbb{R}^O$ that summarizes the \ac{NN}'s local behavior in $\iota$.
The procedure \filterAlg{} then
checks whether a counterexample region is spurious using an SMT solver.
This task is easier than searching nonlinear counterexamples directly since the \ac{NN}'s behavior is summarized by the affine mapping $\omega$.
Using the definitions from the following subsections, this procedure is sound and complete (see proof on \cpageref{proof:nnnv:sound_and_complete}):

\newcommand{\modelsof}[1]{\ensuremath{\text{models}\mleft(#1\mright)}}
\subsection{Linearization}
\label{subsec:overapprox}
\looseness=-1
The procedure \overapproxAlg{} enriches each nonlinear atom $a_i$ of an \ac{OLNNV} query with linear approximations.
The approximations are always with respect to a value range $R$ and we use  overapproximations $\overline{a_i}$ (for any state $\nu$ with $\nu \vDash a_i \land R$ it holds that $\nu \vDash \overline{a_i}$) as well as underapproximations $\underline{a_i}$ (for any state $\nu$ with $\nu \vDash \underline{a_i} \land R$ it holds that $\nu \vDash a_i$).
Essential to this component is the idea that \overapproxAlg{} produces an equivalent formula:
Approximate atoms do not replace, but complement the nonlinear atoms and the generation of concrete linear regions is left to the mosaic step (see \Cref{subsec:query_decompose}).
\overapproxAlg{} is defined as follows:
\begin{definition}[Linearization]
\label{def:lin_overapprox}
\overapproxAlg{} receives an \ac{OLNNV} query $p$ with nonlinear atoms $a_1,\dots,a_k$ and value range $R$ s.t. $p \implies R$ is valid.
It returns a query $p \land \bigwedge_{i=1}^k \left( \left(a_i \implies \overline{a_i}\right) \land \left(\underline{a_i} \implies a_i\right) \right)$ where $\overline{a_i}\in\formulaset{\linreals}$ (resp. $\underline{a_i}\in\formulaset{\linreals}$) are overapproximations (resp. underapproximations) of $a_i$ w.r.t. $R$.
\end{definition}
We use an approximation procedure based on \ac{OVERT}~\cite{Sidrane2021} while further approximating $\max/\min$ terms within \ac{OVERT}.
This results in a disjunction of linear constraints (see below for details).
As highlighted above, \overapproxAlg{} produces equivalent formulas
and therefore retains the relations between linear and nonlinear atoms (see proof on \cpageref{proof:linearization:equivalence}):
\begin{lemmaE}[Equivalence of Linearization][end,text link=]
\label{lem:linearization:equivalence}
Let $p\in\formulaset{\reals}$ be some \ac{OLNNV} query and $p_o$ be the result of $\overapproxAlg{}\mleft(p\mright)$.
Then $p$ is equivalent to $p_o$.
\end{lemmaE}
\begin{proofE}
\label{proof:linearization:equivalence}
Let $p\in\formulaset{\reals}$ be some nonlinear \ac{OLNNV} query and $a_1,\dots,a_k$ be the nonlinear atoms in $p$.
From the definition of \overapproxAlg{}, we know that $p_o$ has the form $p \land \bigwedge_{i=1}^k \left( \left(a_i \implies \overline{a_i}\right) \land \left(\underline{a_i} \implies a_i\right) \right)$.
By definition, it holds for approximations $\underline{a_i}, \overline{a_i}$ that 
for any state $\nu$ with $\nu \vDash a_i \land R$ it also holds that $\nu \vDash \overline{a_i}$ (resp. for any state $\nu$ with $\nu \vDash \underline{a_i} \land R$ it also holds that $\nu \vDash a_i$).
Consequently, the formulas $R \implies \left(a_i \implies \overline{a_i}\right)$ and $R \implies \left(\underline{a_i} \implies a_i\right)$ are valid for all $a_i$.
Let $\nu$ be a state such that $\nu \vDash p$.
Then, by definition $\nu \vDash R$ and due to the above mentioned validity it therefore holds that $\nu \vDash a_i \implies \overline{a_i}$ and $\nu \vDash \underline{a_i} \implies a_i$.
Therefore, $\nu \vDash p_o$.
Conversely, for any state with $\nu \vDash p_o$ it also holds that $\nu \vDash p$.
\end{proofE}

\paragraph{Approximation.}
For concicenes we present our approximation approach for over-approximations.
Our under-approximations are computed in the same manner, however lower and upper bound computation of terms is flipped in this case.
We can approach the question of overapproximation construction from a perspective of models:
For a given formula $\zeta$, let $\dLinterpretationapp{\zeta}=\left\{\nu \in \dLstatespace \middle\mid \nu \vDash \zeta \right\}$ be the set of models (i.e. states satisfying $\zeta$).
We then obtain the following Lemma for the relation between overapproximations and model sets:
\begin{lemma}[Supersets are Overapproximations]
Assume bounds $B$ on all variables and a formula $\zeta$.
Another formula $\zeta_o \in \formulaset{\linreals}$ is a linear overapproximation of $\zeta$ iff $\dLinterpretationapp{B \land \zeta} \subseteq \dLinterpretationapp{B \land \zeta_o}$.
\end{lemma}
This presentation only considers the case of a polynomial constraint $\theta > 0$.
Our approximation procedure begins by computing the relational approximation of $\theta$ using the OVERT algorithm~\cite{Sidrane2021}.
By resolving intermediate variables introduced through OVERT, we obtain an approximation of the form $\underline{\theta_{pwl}} \leq \theta \leq \overline{\theta_{pwl}}$ where both bounds are piece-wise linear functions (i.e. linear real arithmetic with the addition of $\max$ and $\min$ operators).
It then holds that $\dLinterpretationapp{B \land \theta > 0} \subseteq \dLinterpretationapp{B \land \overline{\theta_{pwl}}>0}$.
We now distinguish between univariate and multivariate piece-wise linear behavior:
For univariate piece-wise linear behavior there is some variable $v \in \varsofformula{\overline{\theta_{pwl}}}$ and some coefficient $c \in \reals$ with terms $\theta_1,\theta_2$ such that $\dLinterpretationapp{B \land \overline{\theta_{pwl}}>0} = \dLinterpretationapp{B \land \overline{\theta_1}>0 \land v>c} \cup \dLinterpretationapp{B \land \overline{\theta_2}>0 \land v \leq c}$.
In order to subsume piece-wise linear splits along variable $v$ which are close to $c$, we construct the following overapproximation for a small $\varepsilon>0$ resulting in two normalized queries:
\[
\dLinterpretationapp{B \land \overline{\theta_{pwl}}>0} \subseteq  \dLinterpretationapp{B \land \overline{\theta_1}>0 \land v>\left(c-\varepsilon\right)} \cup \dLinterpretationapp{B \land \overline{\theta_2}>0 \land v \leq \left(c+\varepsilon\right)}.
\]
For the multivariate cases we approximate the piece-wise linear behavior.
In particular, we introduce a (to the best of our knowledge) novel, closed-form upper bound for the linear approximation of $\max$ terms:
\begin{lemma}[Upper Bound for multivariate $\max$]
\label{lemma:approx:upper}
Let $f,g:\mathbb{R}^{I+O}\to\mathbb{R}$ be two linear functions, and let $B\subset\mathbb{R}^{I+O}$ be a closed interval box, then:
Assume ${x_g \coloneqq \arg\max_{x \in B} g\left(x\right) - f\left(x\right)}$ and ${x_f \coloneqq \arg\max_{x \in B} f\left(x\right) - g\left(x\right)}$
 where $f\left(x_f\right)-g\left(x_f\right)$ and $g\left(x_g\right)-f\left(x_f\right)$ are both positive. 
Further, assume the following assignments with ${\gamma \coloneqq f\left(x_f\right) - f\left(x_g\right) - g\left(x_f\right) + g\left(x_g\right)}$:
\begin{align*}
    \mu \coloneqq -\frac{g\left(x_f\right)-f\left(x_f\right)}{\gamma}
    &&
    c \coloneqq -\frac{\left(f\left(x_f\right)-g\left(x_g\right)\right)\left(f\left(x_g\right)-g\left(x_g\right)\right)}{\gamma}.
\end{align*}
In this case, it holds for all $x \in B$ that:
$\mu f\left(x\right) + \left(1-\mu\right) g\left(x\right) + c \geq \max\left(f\left(x\right),g\left(x\right)\right)$.
In particular, it holds that
\[
\dLinterpretationapp{B \land \left(\max\left(f\left(x\right),g\left(x\right)\right) > 0\right)}
\subseteq
\dLinterpretationapp{B \land \left(\left(\mu f\left(x\right) + \left(1-\mu\right) g\left(x\right) + c\right) > 0\right)}
\]
\end{lemma}
\begin{proof}
At first, the choices for $\mu$ and $c$ may seem arbitrary, however they are actually the solution of the following set of equations:
\begin{align*}
    \mu f\left(x_f\right) + \left(1-\mu\right) g\left(x_f\right)+c &= f\left(x_f\right) \\
    \mu f\left(x_g\right) + \left(1-\mu\right) g\left(x_g\right)+c &= g\left(x_g\right)
\end{align*}
The choice of $\mu$ and $c$ ensures that we obtain a shifted convex mixture of the two linear functions that matches $f$ and $g$ at their points of maximal deviation.
We can now prove that this shifted mixture is indeed larger than $f$ or $g$ at any point within $B$.
Let us begin by proving that our bound is larger than $g$ for $x \in B$.
In the following each formula implies the validity of the formula above:
\begin{align*}
    g\left(x\right) &\leq \mu f\left(x\right) + \left(1-\mu\right) g\left(x\right) + c\\
    0 & \leq \mu \left(f\left(x\right) - g\left(x\right)\right) + c\\
    0 & \leq -\frac{1}{\gamma} \left(
    \left(g\left(x_f\right) - f\left(x_f\right) \right)
    \smash{\underbrace{\left(f\left(x\right) - g\left(x\right) \right)}_{
    \mathclap{\leq f\left(x_f\right) - g\left(x_f\right) \text{ for }x \in B}
    }}
    + 
    \left(f\left(x_f\right)-g\left(x_g\right)\right)\left(f\left(x_g\right)-g\left(x_g\right)\right)
    \right)%
    \vphantom{\underbrace{\left(f\left(x\right) - g\left(x\right) \right)}_{
    \mathclap{\leq f\left(x_f\right) - g\left(x_f\right) \text{ for }x \in B}
    }}\\
    0 & \leq -\frac{1}{\gamma} \left(
    \left(g\left(x_f\right) - f\left(x_f\right) \right)
    \left(f\left(x_f\right) - g\left(x_f\right)\right)
    + 
    \left(f\left(x_f\right)-g\left(x_g\right)\right)\left(f\left(x_g\right)-g\left(x_g\right)\right)
    \right) \\
    0 & \leq \frac{1}{\gamma} \left(
    f\left(x_f\right) - g\left(x_f\right)
    \right) \gamma \Leftrightarrow g\left(x_f\right) \leq f\left(x_f\right)
\end{align*}
$g\left(x_f\right) \leq f\left(x_f\right)$ is trivially true since $x_f$ was specifically chosen this way.
We also have to prove that our bound is bigger than $f$ for $x \in B$:
\begin{align*}
f\left(x\right) &\leq \mu f\left(x\right) + \left(1-\mu\right) g\left(x\right) + c\\
0 &\leq \underbrace{
\left(g\left(x\right) - f\left(x\right)\right)
}_{
\mathclap{\geq g\left(x_f\right) - f\left(x_f\right) \text{ for }x \in B}
}
+ 
\underbrace{\mu \left( f\left(x\right)-g\left(x\right) \right) + c}_{
\mathrlap{\geq f\left(x_f\right)-g\left(x_f\right) \text{ (see previous proof)}}
}\\
0 &\leq g\left(x_f\right) - f\left(x_f\right) + f\left(x_f\right) - g\left(x_f\right) = 0
\end{align*}
Thus we obtain an upper bound for the function.
\end{proof}
By applying \ac{OVERT} followed by the univariate resolution and multivariate overapproximation up to saturation, we compute an overapproximation and underapproximation for each nonlinear atom.
Subsequently, we append these new formulas to the original formula as defined in \Cref{def:lin_overapprox}.

\paragraph{Running Example.}
\looseness=-1
The turquoise constraints in \Cref{fig:ncubev} visualize exemplary linearized constraints.
For \ac{ACC} one nonlinear atom is $\accRelPos-\frac{\accRelVel^2}{2B} \geq 0$.
The formula $\text{accApprox} \equiv \accRelPos - \frac{100^2}{2B} \geq 0 \land \accRelVel> 50 \lor \accRelVel\leq 50 \land \accRelPos - \frac{50^2}{2B} \geq 0$
underapproximates the atom for $\accRelVel \in \left[0, 100\right]$.
We can thus append the following formula to our query:
$\left(\text{accApprox} \implies \accRelPos -\frac{\accRelVel^2}{2B} \geq 0\right)$.
For $\accRelVel \in \left[0, 100\right]$ this formula is always satisfied\footnotemark.
\footnotetext{In practice, $\accRelVel$ may also be negative requiring a more complex approximation.}

\subsection{Input Space Mosaics}
\label{subsec:query_decompose}
\looseness=-1
The \normalizeAlg{} procedure takes a central role in the verification of nonlinear, non-normalized \ac{OLNNV} queries. %
Classically, one uses DPLL(T) to decompose an arbitrary formula into conjunctions then handled by a theory solver.
\Ac{OLNNV}'s crux is its use of reachability methods which do not lend themselves well to \textit{classic} DPLL(T):
Its usage would result in duplicate explorations of the same input space w.r.t. different output constraints which is inefficient.
Therefore, we generalize DPLL(T)~\cite{DPLLT} through the \normalizeAlg{} procedure.
The procedure receives a quantifier-free\footnote{A quantifier-free $p$ can be assumed as real arithmetic admits quantifier elimination~\cite{tarski1951}. In practice, all queries of interest were already quantifier-free.}, non-normalized \ac{OLNNV} query and 
enumerates \emph{azulejos} of the input space each with an associated normalized linear \ac{OLNNV} query $q_l$ (conjunction over input atoms, disjunctive normal form over output atoms) and nonlinear atoms in disjunctive normal form $q_n$.
The input space is thus turned into a mosaic and the disjunction over all queries is \emph{equivalent} to the input query.
We can then obtain classical DPLL(T) by marking all atoms as linear input constraints.
Our implementation of \normalizeAlg{} instruments a SAT solver on the Boolean skeleton of $p$ as well as a real arithmetic SMT solver to restructure a formula in this way. %

For a formula $\zeta \in \formulaset{\reals}$, let $\satatoms{\zeta}$ be the set of set of signed atoms such that for all $A \in \satatoms{\zeta}$
it holds that $A$ only contains atoms of $\zeta$ or its negations ($A \subseteq \atomsof{\zeta} \cup \left\{\neg b\mid b \in\atomsof{\zeta}\right\}$).
Further, we require for $\satatoms{\zeta}$ that for any state $\nu$ it holds that $\nu \vDash \zeta$ iff there exists an $A \in \satatoms{\zeta}$ such that $\nu \vDash \bigwedge_{a \in A} a$.
Note that there may exist multiple such sets in which case we can choose an arbitrary one.
For example, for $\zeta \equiv x>0 \lor \neg\left(y>0\right)$ we could get 
$
\satatoms{\zeta} = \left\{
\left\{x>0\right\},
\left\{\neg\left(x>0\right),\neg\left(y>0\right)\right\}\right\}
$.
For a given formula $\zeta$, we will call $J_\zeta$ the set of input variables.
We introduce the following notation for projection of \text{sat-atoms} on the set $J_\zeta$:
\[
\projection{\satatoms{\zeta}}{J_\zeta} = \left\{ \left\{ a \mid a \in A \land \varsofformula{a} \subseteq J_\zeta\right\} \mid A \in \satatoms{\zeta} \right\}.
\]
For example, reconsidering the previous example with $J_\zeta=\left\{x\right\}$ we would get 
$\projection{\satatoms{\zeta}}{J_\zeta} = \left\{\left\{ x > 0\right\}, \left\{\neg\left(x>0\right)\right\}\right\}$.
For a given \ac{OLNNV} query $p_o$ and a given set of input variables $J_{p_o}$, \normalizeAlg{} then initially enumerates feasible combinations of linear input atoms (the azulejos) and for each such combination feasible combinations of mixed/output atoms are enumerated.
This results in the following set:
\begin{equation*}
S_1 = \left\{
\left(
\bigwedge_{\substack{a\in i\\a \in \formulaset{\linreals}}} a
\right)
\land \left(\bigvee_{
o \in \satatoms{p_o \land i}
} \left(\bigwedge_{\substack{b \in o\\b \in \formulaset{\linreals}}} b\right)\right)
~\middle|~
i \in \projection{\satatoms{p_o}}{J_{p_o}}
\right\}
\end{equation*}
Finally, for each $q_l \in S_1$, we can generate all possible combinations of nonlinear atoms $S_2=\satatoms{p_o\land q_l}$ and generate their disjunction:
\begin{equation}
    q_n \equiv \bigvee_{A \in S_2} \bigwedge_{a \in A} a
\end{equation}
We achieve this by enumerating all satisfying assignments for the boolean skeleton of $p_o$ (i.e. the formula where all atoms are substituted by boolean variables) using an incremental SAT solver.
This initially happens in the same manner as it is done for the classical version of DPLL(T).
However, once a model is found, we fix the assignment of linear input-only atoms and enumerate all other satisfying assignments of linear atoms, generating the disjunctions within $S_1$.
For each conjunction we additionally enumerate possible assignments for the nonlinear atoms.
Notably, through the encoding of \overapproxAlg{} the procedure automatically knows which truth-combinations of a nonlinear constraint and its approximations may appear.
We additionally provide information on linear dependencies between linear atoms to the SAT solver.
All enumeration procedures are interleaved with calls to SMT solvers for linear and polynomial real arithmetic constraints, which check whether a given combination of constraints is indeed also satisfiable in the theory of real arithmetic (i.e. when interpreting the atoms as real arithmetic constraints instead of as boolean variables).
In order to discard unsatisfiable solutions more quickly, we make use of unsatisfiability cores and conversely use a cache for satisfiable assignment combinations.
We exploit partial models returned by the SAT solver to omit atoms which can (potentially) appear in both polarities for a given combination of constraints.

We show that the decomposition is correct (i.e. the disjunction over all queries is equivalent to the original query, see \Cref{lem:decomp:correct}) and that it is minimal in the sense that 
the resulting azulejos do not overlap (see \Cref{lem:decomp:minimum}) with proofs on \cpageref{proof:decomp:correct,proof:decomp:minimum}:
\begin{propositionE}[Correctness of Mosaic][end,text link=]
\label{lem:decomp:correct}
Let $p$ be any \ac{OLNNV} query.
Let $Q \subset \formulaset{\linreals}\times\formulaset{\reals}$ be the set returned by $\normalizeAlg{}\mleft(p\mright)$, then the following formula is valid:
$p \iff \left(\bigvee_{\left(q_l,q_n\right) \in Q} \left(q_l \land q_n\right)\right)$
\end{propositionE}%
\begin{proofE}
\label{proof:decomp:correct}
We begin by considering the case where some state $\nu$ satisfies $\bigvee_{\left(q_l,q_n\right) \in Q} q_l \land q_n$.
By definition, this means that there exists some $\left(q^*_l, q^*_n\right) \in Q$ such that $\nu \vDash q_l \land q_n$.
Through the definition of the set $S_1$ in \Cref{apx:mosaic}, we know that $q^*_l$ contains a conjunction over linear input atoms $i^*_l$.
Let $o^*_l \in \satatoms{p \land i^*_l}$ be the set of mixed/output atoms such that $\nu \vDash \bigwedge_{b \in o^*} b$.
Further, since $\nu \vDash q_n$, we know there also exists an $A^*\in\satatoms{p\land q_l}$ such that $\nu \vDash \bigwedge_{a^*\in A^*}a^*$.
Through the definition of sat-atoms and its projection we then know that $A^* \cup i^* \cup o^* \in \satatoms{p}$.
Consequently, it must hold that $\nu \vDash p$.

Consider now the other direction where for some state $\nu$ it holds that $\nu \vDash p$.
By definition of sat-atoms, its projection and $S_1$ we know that there must exist some $i^* \in \projection{\satatoms{p}}{J_{p}}$ such that $\nu \vDash i^*$.
Moreover, there must exist an $o \in \satatoms{p \land i^*}$ such that $\nu \vDash \bigwedge_{b \in o^*} b$.
Finally, since $\nu \vDash  i^* \land o$, there must exist an $A^* \in \satatoms{p \land i^* \land o}$ such that $\nu \vDash  \bigwedge_{a^*\in A^*}a^*$
Consequently, there exists a $\left(q_l,q_n\right) \in Q$ such that $\nu \vDash q_l \land q_n$ and therefore $\nu \vDash \bigvee_{\left(q_l,q_n\right) \in Q} q_l \land q_n$.
\end{proofE}%
\begin{propositionE}[Flatness of Mosaic][end,text link=]
\label{lem:decomp:minimum}
\looseness=-1
Let $\left(i_1 \land \bigvee_j o_{1,j}\right), \left( i_2 \land \bigvee_j o_{2,j} \right)$ be two linear queries enumerated by \normalizeAlg{} then $i_1 \land i_2$ is unsatisfiable.
\end{propositionE}%
\begin{proofE}
\label{proof:decomp:minimum}
Assume there were two linear queries $\left(i_1 \land \bigvee_j o_{1,j}\right)$ and $\left( i_2 \land \bigvee_j o_{2,j} \right)$ such that $i_1 \land i_2$ had a model.
By definition, each set $A \in \satatoms{p}$ must contain each atom of $p$ or its negation.
Consider now the projection $\projection{\satatoms{p}}{J_{p}}$ from which we obtain all $i$s (in particular $i_1$ and $i_2$):
Since $i_1$ and $i_2$ contain the same set of atoms, it must be the case that for some atom $a \in i_1$, it holds that $\neg a \in i_2$ or vice versa (otherwise, the two would be identical).
Through the law of the excluded middle, we get that $a \land \neg a$ is unsatisfiable, and thus $i_1 \land i_2$ is unsatisfiable.
\end{proofE}%
\noindent
We could use this approach to decompose a nonlinear formula into a set of normalized linear \ac{OLNNV} queries without approximation.
\normalizeAlg{} then soundly omits all nonlinear constraints.
However, this leads to many spurious counterexamples.
Therefore, we add linear approximations (\Cref{subsec:overapprox}) of atoms which are then automatically part of the conjunctions returned by \normalizeAlg{}.

\paragraph{Running Example.}
\looseness=-1
In the previous section we extended our query by a linear underapproximation.
Our procedure generates an azulejo for the case where $\accRelPos - \frac{100^2}{2B} \geq 0 \land \accRelVel > 50$ is satisfied (implying $\accRelPos - \frac{\accRelVel^2}{2B} \geq 0$) and the case where it is not.
While the linear approximation is an edge of the mosaic tile, the original atom (the tile's ``painting'' describing the precise constraint) would be part of the \emph{nonlinear} disjunctive normal form.
For each azulejo, the output conjunctions of $\text{accCtrlFml}$ are enumerated.
For OVERT's approximation with $N=1$, our implementation decomposes the ACC query into 20 normalized queries with up to 10 cases in the output constraint disjunction.
Without \normalizeAlg{} each case would be treated as a separate reachability query leading to significant duplicate work.

\paragraph{Relation to DPLL(T)}
Abstracting away the real-arithmetic, the \normalizeAlg{} algorithm generates tuples of normalized \ac{OLNNV} queries and disjunctive normal forms that are satisfiable w.r.t. a theory solver $T$.
The algorithm itself interleaves SAT-based reasoning about a boolean abstraction (annotated with information on whether an atom is linear and/or an input constraint) and theory solver invocations.
We can now consider the case where all atoms (independent of their concrete contents and the theory $T$) are annotated as linear input constraints:
In this case \normalizeAlg{} merely returns a mosaic of this ``input'' space where each azulejo corresponds to a conjunction of atoms that is satisfiable w.r.t. to the theory solver $T$ and the disjunction over all those conjunctions is then once again equivalent to the original formula -- corresponding to DPLL(T)'s behavior.

\subsection{Counterexample Generalization and Enumeration}
\label{subsec:enumerate}
The innermost component of our algorithm enumerates all counterexample regions (\enumerateAlg{}).
To this end, \enumerateAlg{} requires an algorithm which generalizes counterexample \emph{points} returned by \ac{OLNNV} to \emph{regions} (\generalizeAlg{}).
For each such counterexample region we can then check if there exist concrete violations of the nonlinear constraints (\filterAlg{}).
We begin by explaining \generalizeAlg{} which
converts a counterexample point returned by \ac{OLNNV} into a counterexample region.
The key insight for this approach is that a concrete counterexample $\left(z_0,x_0\right)$ 
returned by an \ac{OLNNV} tool
induces a region of points with \emph{similar behavior} in the \ac{NN}.
A concrete input $z_0$ induces a fixed activation pattern for all piece-wise linear activations within the \ac{NN} in a region $\iota$ around $z_0$.
Consider the first layer's activation function $f^{(1)}$:
$f^{(1)}$ can be decomposed into linear functions $f_i$ and so is a sum of affine transformations $A_i z_0 + b_i$ which are active iff $q_i\left(z_0\right)$ is true.
We can then describe $f^{(1)}$'s local behavior around $z_0$ as the linear combination of all affine transformations active for $z_0$.
This sum is itself an affine transformation.
By iterating this approach across layers, we obtain a single affine transformation $\omega$ describing the \ac{NN}'s behavior in $\iota$.
The regions returned by \generalizeAlg{} are then defined as follows:
\begin{definition}[Counterexample Region]
For a given \ac{OLNNV} query $q$ and piece-wise linear \ac{NN} $g$, let $\left(z_0,x_0\right) \in \reals^{I}\times\reals^{O}$ be a counterexample, i.e. $x_0=g\mleft(z_0\mright)$ and $q\mleft(z_0,x_0\mright)$ holds.
The \emph{counterexample region} for $z_0$ is the maximal polytope $\iota \subset \reals^I$ with a linear function $\omega$  s.t. $z_0 \in \iota$ and $\omega\mleft(z\mright)=g\mleft(z\mright)$ for all $z \in \iota$.
\end{definition}
Star Sets~\cite{tran2019star,bak2020improved} can compute $\left(\iota,\omega\right)$ by steering the Star Set according to the activations of $z_0$.
As the number of counterexample regions is exponentially bounded by the number of piece-wise linear nodes, we can use \generalizeAlg{} for exhaustive enumeration.
This is only a worst-case bound due to the NP-completeness of \ac{NN} verification~\cite{katz2017reluplex,DBLP:conf/rp/SalzerL21}.
In practice, the number of regions is much lower since many activation functions are linear in all considered states.
While a given counterexample region certainly has a point violating the \emph{linear} query that was given to the \ac{OLNNV} tool, it may be the case that the counterexample is \emph{spurious}, i.e. it does not violate the nonlinear constraints.
However, we can use the concise description of counterexample regions to check whether this is the case:
The function $\omega$ describes the \ac{NN}'s \emph{entire} behavior within $\iota$ as a single affine transformation and is thus much better suited for \acs{SMT}-based reasoning.
This \ac{SMT}-based check is performed by \filterAlg{} based on the following insight:
\begin{lemmaE}[Counterexample Filter][end]
\label{lem:counter-example-filtering}
Let $\left(q_l,q_n\right)$ be a tuple returned by \normalizeAlg{}.
A counterexample region $\left(\iota,\omega\right)$ for $q_l$ is a counterexample region for $q_l \land q_n$ iff the formula 
$
\eta \equiv \left( q_l\mleft(z,x^+\mright) \land q_n\mleft(z,x^+\mright) \land z \in \iota \land x^+ = \omega\mleft(z\mright) \right)
$
is satisfiable.
\end{lemmaE}
\begin{proofE}
Assume some $\left(\iota,\omega\right)$ is indeed a counterexample region for $q_l \land q_n$.
In this case, we know that there is some $z\in\iota$ such that with $x^+=g\mleft(z\mright)$ we get $q_l\mleft(z,x^+\mright) \land q_n\mleft(z,x^+\mright)$.
However, by definition of counterexample regions we also know that $g\mleft(z\mright)=\omega\mleft(z\mright)$.
Therefore, the assignments of $z$ and $x^+$ satisfy $\eta$.
Next, consider the other direction.
I.e. we assume we have a satisfying assignment for $\eta$.
By definition we know that for the given assignment of $z$ it holds that $x^+=g\mleft(z\mright)=\omega\mleft(z\mright)$.
Therefore, $z,x^+$ respect the neural network and satisfy $q_l \land q_n$, which are the two requirements for a counterexample.
\end{proofE}
The size of the formula $\eta$ only depends on $q_l$, $q_n$, $I$, and $O$ and, crucially, is \emph{independent} of the size and architecture of the \ac{NN}.
In practice, even $x^+$ can be eliminated (substitute linear terms of $\omega\left(x^+\right)$).

Based on these insights, the last required component is a mechanism for the exhaustive enumeration of all counterexample regions (denoted as $\enumerateAlg{}$).
There are two options for \enumerateAlg{}:
Either we use geometric path enumeration~\cite{tran2019star,bak2020improved} to enumerate all counterexample regions (used for the evaluation) or we instrument arbitrary complete off-the-shelf \ac{OLNNV} tools for linear queries through \Cref{algorithm:prep:extension}.
We define \enumerateAlg{} as follows:
\begin{definition}[Exhaustive Counterexample Generation]
\label{def:enumerateAlg}
An exhaustive enumeration procedure $\enumerateAlg{}$ receives a linear, normalized \ac{OLNNV} query $q$ and a piece-wise linear \ac{NN} $g$ and returns a covering $E$ of counterexample regions, i.e. $E$ satisfies
$
\left\{z \in \reals^I \mid q\mleft(z,g\mleft(z\mright)\mright) \right\} \subseteq \bigcup_{\left(\iota,\omega\right) \in E} \iota.
$
\end{definition}

\begin{algorithm}
    \caption{Enumeration of counterexample regions using off-the-shelf \ac{OLNNV} tools.}
    \label{algorithm:prep:extension}
    \begin{algorithmic}
    \Require Query $\bar{p}$, \ac{FNN} $g$
    \Procedure{Enumerate}{$\bar{p},g$}
    \State $s,E$ $\leftarrow$ \texttt{sat}$, \emptyset$
    
    \While{$s=\texttt{sat}$}
        \State $s, e \leftarrow \Call{NNV}{\bar{p},g}$
        \Comment{Call \ac{OLNNV} tool}
        \If{$s=\texttt{sat}$}
            \State $\iota,\omega \leftarrow \Call{Generalize}{e,g}$
            \Comment{Generalize counterexample}
            \State $E \leftarrow E \cup \left\{\left(\iota,\omega\right)\right\}$
            \Comment{Store counterexample}
            \State $\bar{p} \leftarrow \bar{p} \land \neg \iota$
            \Comment{Exclude counterexample region from remaining search space}
        \EndIf
    \EndWhile
    \State \Return $E$
    \EndProcedure
    \end{algorithmic}
\end{algorithm}

\section{Proofs}
\label{apx:proofs}
\printProofs

\section{Adaptive Cruise Control}
\label{apx:acc}
\paragraph{Information on the \ac{dL} model.}
The controller $\alpha_{\text{ctrl}}$ has three nondeterministic options:
it can brake with $-B$ (no constraints), set relative acceleration to $\accRelAcc = 0$ (constraint $\text{accCtrl}_0$) or choose any value in the range $\left[-B,A\right]$ (constraint $\text{accCtrl}_1$).
The constraints for the second and third action are as follows:
\begin{align*}
    \text{accCtrl}_0 \equiv &\left(2B\left(\accRelPos+T\accRelVel\right)>\accRelVel^2\right)\\
    \text{accCtrl}_1 \equiv &
    2B\left( \accRelPos + T\accRelVel + 0.5T^2\accRelAcc  \right) > \left(\accRelVel + T\accRelAcc\right)^2 \land \\
    &\left(
            -\accRelVel > T\accRelAcc \lor 0 < \accRelVel
         \lor \left( \accRelVel^2 < 2\accRelAcc\accRelPos \right)
        \right)
\end{align*}
We can prove the safety of this control envelope for the following initial condition which is also the loop invariant:
\[
\text{accInit} \equiv \text{accInv} \equiv \accRelPos > 0 \land \accRelPos 2 B \geq \accRelVel^2
\]
The right-hand side of the invariant/initial condition ensures that the distance is still large enough to avoid a collision through an emergency brake ($\accRelAcc = -B$).
Based on these foundations, the full specification for the \ac{NN} generated by \chThreeFancyName{} reads as follows:
\begin{align*}
        \Big(&0\leq \accRelPos \land \accRelPos \leq 100 \land -200 \leq \accRelVel \land \accRelVel \leq 200~\land\\
        &-B \leq \accRelAcc^+ \land \accRelAcc^+ \leq A~\land\\
        &\accRelPos > 0 \land \accRelPos \geq \accRelVel^2/(2*B)\Big)~\implies\\
        \Big(&\accRelAcc^+ \geq A = \accRelVel~\lor\\
        &\accRelAcc^+ \geq -B \land
        \accRelAcc^+  <  A \land
        \accRelAcc^+ \neq 0~\land\\
        &\big(
        (-\accRelVel/\accRelAcc^+  > T \lor -\accRelVel  <  0)~\land\\
        &\accRelPos + \accRelVel * T + \accRelAcc^+ * T^2 / 2 > (\accRelVel + \accRelAcc^+ * T)^2 / (2 * B) ~\lor\\
        &\accRelPos + \accRelVel * T + \accRelAcc^+ * T^2 / 2 > (\accRelVel + \accRelAcc^+ * T)^2 / (2 * B) ~\land\\
        &\accRelPos*\accRelAcc^+ - \accRelVel^2 + \accRelVel^2 /2 > 0
        \big)~\lor\\
        &\accRelPos + \accRelVel * T > \accRelVel^2 / (2 * B)~\land\accRelAcc^+ = 0 
        \Big)
\end{align*}
For our verification, we set $T=0.1$ (note that this is a bound on the frequency of control decisions, not a time horizon) and $A=B=100$.

\section{Extended Evaluation}
\label{apx:evaluation}
We implemented our procedure in a new tool\sepfootnote{SNNT} called \ac{SNNT}.
Due to the widespread use of ReLU \acp{NN}, \ac{SNNT} focuses on the verification of generic \ac{OLNNV} queries for such \acp{NN}, but could be extended in future work.
Our tool is implemented in Julia~\cite{Bezanson_Julia_A_fresh_2017} using  nnenum~\cite{bak2020improved,bakOverapprox} for \ac{OLNNV}, PicoSAT~\cite{PicoSAT,PicoSatJL} and Z3~\cite{Z3overall,Z3nonlinear}.
Our evaluation aimed at answering the following questions:
\begin{enumerate}
    \item[Q1] Can \ac{SNNT} verify infinite-time horizon safety or exhaustively enumerate counterexample regions for a given \ac{NNCS}?
    \item[Q2] Does our approach advance the State-of-the-Art?
    \item[Q3] Does our approach scale to complex real-world scenarios such as ACAS X?
\end{enumerate}
The case studies comprised continuous and discrete control outputs.
(Q3) is answered in the paper's main evaluation (see \Cref{sec:evaluation}; the remaining questions are discussed below.
Times are wall-clock times on a 16 core AMD Ryzen 7 PRO 5850U CPU (\ac{SNNT} itself is sequential while nnenum uses multithreading).

\subsection{Verification of Adaptive Cruise Control}
\label{subsec:eval:acc}
\looseness=-1
\looseness=-1
We applied our approach to the previously outlined running example.
To this end, we trained two \acp{NN} using PPO~\cite{Raffin2021}: \texttt{ACC} contains 2 layers with 64 \relu{} nodes each while \texttt{ACC\_Large} contains 4 layers with 64 \relu{} nodes each.
Our approach only analyzes the hybrid system \emph{once} and reuses the formulas for all future verification tasks (e.g. after retraining). %
We analyzed both \acp{NN} and a third one (see below for details) using \ac{SNNT} for coarser and tighter approximation settings (using OVERT's setting $N\in\left\{1,2,3\right\}$) on the value range $\left(\accRelPos,\accRelVel\right)\in\left[0,100\right]\times\left[-200,200\right]$.
The analyses took 47 to 300 seconds depending on the \ac{NN} and approximation.
The runtimes show mixed results for tighter approximations:
While tighter approximations (i.e. a higher $N$) sometimes improves performance (e.g. for \texttt{ACC\_Large} retrained), it can also harm performance (e.g.  as seen for \texttt{ACC\_Large}).
We suspect that this is a combination of two factors.
First, finer approximations yield a larger number of queries which may increase the overall overhead.
Secondly, our adjustments to OVERT's approximation using approximation of piece-wise linearities (see \Cref{subsec:overapprox}) may in some cases worsen the approximation in comparison to a lower $N$.
We leave a more fine-grained analysis of approximation techniques to future work and focus our analysis on approximations with $N=1$.
Across all \acp{NN} we find that approximation helps:
The first row shows performance when omitting the approximated constraints in the \ac{OLNNV} query and uniformly performs worse than an $N=1$ approximation.
\ac{SNNT} finds the \ac{NN} \texttt{ACC\_Large} to be \emph{unsafe} and provides an exhaustive characterization of all input space regions with unsafe actions.

\paragraph{Information on the counterexamples found for \texttt{ACC\_Large}.}
\Cref{fig:acc_larger_trajectory} shows the input state where the x-axis represents possible values for $\accRelPos$ and the y-axis represents possible values for $\accRelVel$.
The orange line represents the edge of the safe state space, i.e. all values below the orange line are outside the reachable state space of the contract.
The red areas represent all parts of the state space where \ac{SNNT} found concrete counterexamples for the checked controller monitor formula.
Furthermore, the plot contains two lines representing the system's evolution over time when started at certain initial states.
In particular, we observe one trajectory leading to a crash due to an erroneous decision in the red area around $\accRelPos=5,\accRelVel=-25$.
This concrete counterexample was found by sampling initial states from the regions provided by \ac{SNNT}.
\begin{figure}
    \centering
    \includegraphics[width=0.85\textwidth]{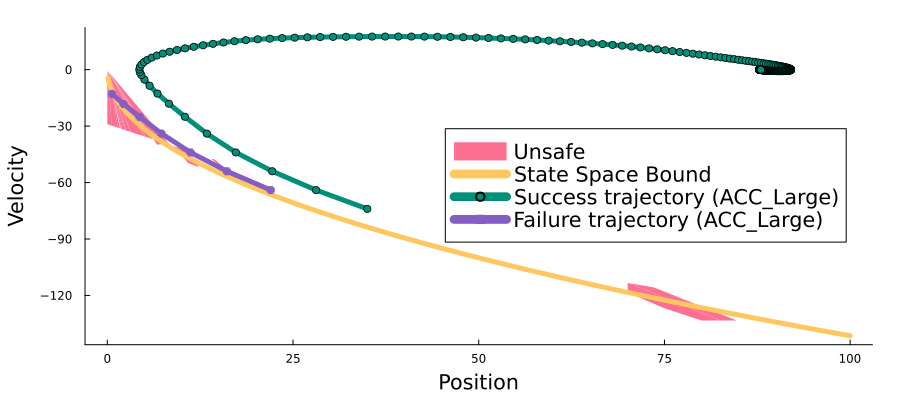}
    \caption{This plot shows the (input) state space for the \texttt{ACC\_Large} \ac{NNCS}: The orange line represents the boundary of the safe state space; the red areas indicate regions with counterexamples; the purple and green lines show potential trajectories of the system (dots represent discrete controller decisions).}
    \label{fig:acc_larger_trajectory}
\end{figure}

\paragraph{Information on runtimes.}
The runtimes can be seen in \Cref{tab:accruntimes} where \#Filtered corresponds to the number of counterexample regions that were found to be spurious for \texttt{ACC\_Large}.
Comparing linear and $N=1$ performance we see that approximation especially helps for larger \acp{NN} (\texttt{ACC} vs \texttt{ACC\_Large}).
This is the case because the overapproximate constraints can filter out numerous counterexample regions, which would otherwise have to be processed by the \filterAlg{} procedure.
This effect is less significant for smaller networks where the time for overapproximation construction takes longer in comparison to the \ac{NN} analysis time and less counterexamples are generated due to the lower number of ReLU nodes.
\begin{table}
    \caption{Runtime of \ac{SNNT} on the ACC networks per approximation.
    Final column lists filtered counterexamples for \texttt{ACC\_Large}}
    \begin{tabularx}{0.7\textwidth}{X|c|c|c|c}
    Approx. &\texttt{ACC} & \texttt{ACC\_Large} & \texttt{ACC\_Large} retrained & \#Filtered\\\hline
    linear & 65s & 228s & 277s & 5658 \\\hline
    $N=1$ & 47s & 87s & 124s & 1889\\\hline
    $N=2$ & 53s & 300s& \hphantom{1}70s & 5604\\\hline
    $N=3$ & 92s & 139s& 102s& 1306\\
    \end{tabularx}
    \label{tab:accruntimes}
\end{table}

\paragraph{Further Training.}
\looseness=-1
Approx. 3\% of \texttt{ACC\_Large}'s inputs resulted in unsafe actions which demonstrably resulted in car crashes.
We performed a second training round on \texttt{ACC\_Large} where we initialized the system within the counterexample regions for a boosted $p \approx 13\%$ of all runs (choosing the best-performing $p$).
By iterating this approach twice, we obtained an \ac{NN} which was safe 
except for very small relative distances 
(for $\left(\accRelPos,\accRelVel\right)\in\left[0,0.08\right]\times\left[-2,0.1\right]$).
\ac{SNNT} certifies the safety outside this remaining region (see column \texttt{ACC\_Large} retrained) which can be safeguarded using an emergency braking backup controller.
Notably, this an \emph{a priori} guarantee is for an arbitrarily long trip.

\paragraph{Results on Q1.}
Our tool \ac{SNNT} is capable of verifying and refuting infinite-time horizon safety for a given \ac{dL} contract.
The support for exhaustively enumerating counterexamples can help in guiding the development of safer \acp{NNCS}.

\subsection{Comparison to Other Techniques}%
\label{subsec:eval:closed_loop}
\looseness=-1
Although \Ac{CLNNV} tools focus on finite-time horizons,
we did compare our approach with the tools from ARCH Comp 2022~\cite{lopez2022arch} (a superset of ARCH Comp 2023~\cite{DBLP:conf/arch/LopezAFJL023}) on \ac{ACC}.
We began by evaluating safety certification on a small subset of the input space of \texttt{ACC\_Large} ($0.009\%$ of the states verified by \ac{SNNT})
for multiple configurations of each tool (see \Cref{tab:clnnv_comp}).
Only \texttt{NNV} was capable of showing safety for 0.1 seconds (vs. time unbounded safety) while taking vastly longer for the tiny fraction of the state space.

\begin{table}[t]
    \centering
    \caption{Comparison of verification tools for \acp{NNCS} on the ACC Benchmark:
    Share of state space analyzed and best results of each tool.}
    \begin{tabular}{l|l|l|l|p{1.8cm}|l}
        Tool & Nonlinearities & Evaluated & Time (s) & Share of & Result\\
        &&Configurations&&State Space&\\\hline
        NNV\hfill~\cite{tran2019star,tran2020neural,9093970} & no & \hfill 4 & \hfill 711 & \hphantom{10}0.009\% & \textbf{safe for 0.1s}\\
        JuliaReach\hfill~\cite{Forets2019JuliaReachReachability,Schilling2021} & no & \hfill 4 & \hfill --- & \hphantom{10}0.009\% & unknown\\ %
        CORA\hfill~\cite{althoff2015introduction} & yes & \hfill 10 & \hfill --- & \hphantom{10}0.009\% & unknown\\ %
        POLAR\hfill~\cite{Polar2022} & poly. Zono. & \hfill $12$ & \hfill --- & \hphantom{10}0.009\% & unknown\\\hline%
        \textbf{\ac{SNNT}} & polynomial & \hfill 1 & \hfill \textbf{124} & 100.000\% & \textbf{safe for $\infty$}
    \end{tabular}
    \label{tab:clnnv_comp}
\end{table}

\paragraph{Comparison to \texttt{NNV}.}
\looseness=-1
We performed a more extensive comparison to \texttt{NNV} by attempting to prove with \texttt{NNV} that the \ac{NNCS} has no trajectories leading from within to outside the loop invariant.
This would witness infinite-time safety.
Due to a lack of support for nonlinear constraints, we approximate the regions.
Over-approximating the invariant as an input region trivially produces unsafe trajectories, thus we can only under-approximate.
Notably, this immediately upends any soundness or completeness guarantees (it does not consider all possible \ac{NN} inputs nor all allowed actions).
We apply an interval-based approximation scheme similar to OVERT (detailed in the subsequent paragraph).
This scheme is parameterized by $\accRelPos$'s step size ($\sigma$), $\accRelVel$'s distance to the invariant ($\varepsilon$) and the step size for approximating the unsafe set ($\rho$).
The right configuration of $\left(\sigma,\varepsilon,\rho\right)$ is highly influential, but equally unclear.
For example, with $\sigma=0.25,\epsilon=5$ and $\rho=1$
we can ``verify''
not only the retrained \texttt{ACC\_Large} \ac{NN} for $2 \leq \accRelPos \leq 3$, but also the original, unsafe \texttt{ACC\_Large} \emph{despite concrete counterexamples}.
This is a consequence of a coarse approximation, but also a symptom of a larger problem:
Neither over- nor under-approximation yields useful results.
In particular, discarding inputs close to the invariant's edge equally removes states most prone to unsafe behavior (see \Cref{fig:acc_larger_trajectory} in \Cref{subsec:eval:acc}).

\paragraph{Approximation Scheme employed for NNV.}
We consider input space stripes of width $\sigma>0$, i.e.  $\left(\accRelPos,\accRelVel\right)\in\left[p_0,p_0+\sigma\right]\times\left[-\sqrt{2A\accRelPos}+\epsilon,\frac{TB}{2}\right]$ ($\accRelVel$ is bounded through the minimally allowed velocity and the maximal velocity that can still decrease $\accRelPos$).
While $\sigma$ determines the granularity of the under-approximation of \ac{NN} inputs, $\epsilon>0$ discards velocities too close to the loop invariant which cannot be proven using an underapproximation and must be non-zero to prove any system.
For each stripe we compute the smallest reachable position $p^*$ and compute a piece-wise linear overapproximation of the negated loop invariant $\accRelPos < \accRelVel^2/\left(2A\right)$ on the interval $\left[p^*,p_0\right]$ using an approach conceptually similar to OVERT~\cite{Sidrane2021}.
We determine the number of pieces through a step size $\rho>0$, i.e. the interval $\left[p^*,p^*+\rho\right]$ will have a different line segment than the interval $\left[p^*+\rho,p^*+2\rho\right]$.
As the negation of the loop invariant (the orange line in \Cref{fig:acc_larger_trajectory}) is non-convex, we integrated an iterative check for disjunctions of unsafe sets into the verification procedure of \texttt{NNV}.

\paragraph{Comparison with DNNV (\Cref{tab:dnnv_comp})}
As DNNV~\cite{Shriver2021} does not support nonlinear properties, a direct comparison to the tool is impossible.
However, we improve upon one important feature implemented in DNNV, namely query normalization.
By exporting the boolean skeleton generated by \normalizeAlg{}, we can use projected model counting~\cite{SRSM19} to estimate the number of propositionally satisfiable conjunctions over linear constraints.
Although the rule based normalization performed by DNNV may produce fewer formulas (this depends on the formula structure and implementation details of the rewriting system),
this count provides an upper bound on the number of conjunctions that can be generated for a formula.
Without \normalizeAlg{}, a rewriting system would first generate a large disjunctive normal form (with at most \#Conjunction many elements), then check the feasibility of generated conjunctions and hand feasible conjunctions to an \ac{OLNNV} tool.
As indicated by \Cref{tab:dnnv_comp}, such an approach can lead to the number of conjunctions reaching into the trillions which becomes entirely intractable in practice.
As can be seen in \Cref{tab:dnnv_comp}, our tool (\# Queries) only produces a fraction of the propositionally satisfiable conjunctions (\# Conjunctions) and also significantly reduces the number of \ac{OLNNV} queries in comparison to an approach that splits up disjunctions (\# Feasible Conjunctions).
Note, that DNNV is also required to check generated conjunctions for feasibility, thus, our approach is also efficient in this regard by requiring a comparatively low number of SMT calls.
Given the feasiblity of 39 trillion conjunctions, one may wonder whether the propositional structure encoded in the boolean skeleton is of use at all.
In this instance, we consider conjunctions over 110 distinct atoms.
Indeed, the propositional structure adds value: Without it, we would obtain  $2^{110}\approx 10^{33}$ possible conjunctions, i.e. based on the propositional structure we only consider a fraction of approx. $10^{-19}$ of all possible combinations. 
\normalizeAlg{} further reduces this fraction to a degree that is manageable via \ac{OLNNV}.

\begin{table}
    \centering
    \begin{tabular}{l|l|l|l||l}
        Property & \# Conjunctions & \textbf{\# Queries}& \# Feasible Conjunctions & \# SMT calls  \\\hline
        ACC & \hphantom{99}2.4k & 20\hphantom{.9K} & 86\hphantom{.9K} & 261\hphantom{.9K}\\
        ACC (Fallback) & \hphantom{99}5.1k & 15\hphantom{.9K} & 72\hphantom{.9K} & 235\hphantom{.9K}\\\hline
        ACAS (DNC) & 117.5M & \hphantom{9}1.7k & \hphantom{9}9.9k & \hphantom{9}11.4k\\
        ACAS (DND) & \hphantom{9}88.9M & \hphantom{9}1.8k & 10.4k & \hphantom{9}12.0k\\
        ACAS (DES1500) & 451.3B & 12.5k & 58.8k & \hphantom{9}66.4k\\
        ACAS (CLI1500) & 374.4B & 13.1k & 62.5k & \hphantom{9}70.4k\\
        ACAS (SDES1500) & \hphantom{99}9.1T & 18.6k & 64.1k & 75.8k\\ %
        ACAS (SCLI1500) & \hphantom{9}18.2T & 21.8k & 76.0k & 88.5k\\
        ACAS (SDES2500) & \hphantom{9}39.0T & 19.0k & 66.7k & 78.5k\\
        ACAS (SCLI2500) & \hphantom{9}19.4T & 18.6k & 67.7k & 79.8k\\
    \end{tabular}
    \caption{Comparison of feasible conjunctions/queries for non-normalized \ac{OLNNV} queries for an approximation with $N=1$: \#Conjunctions is the number of propositionally satisfiable conjunctions over linear constraints, \textbf{\# Queries is the number of \ac{OLNNV} queries generated by \ac{SNNT}}, \# Feasible Conjunctions is the number of \ac{OLNNV} queries when splitting up disjunctions, \# SMT calls is the number of feasibility checks performed by \ac{SNNT}'s \normalizeAlg{} implementation during query generation.}
    \label{tab:dnnv_comp}
\end{table}

\paragraph{Comparison with SMT solvers (\Cref{tab:smt_comp})}
An alternative approach for the verification of non-linear \ac{OLNNV} queries could be encoding the problem using an off-the-shelf SMT solver.
In this case, the SMT solver has to check the satisfiability of the nonlinear \Cref{eq:safety_property}.
We can instrument the Lantern package~\cite{Genin2022} to encode the \ac{NN} into a SMT formula.
Thus, we performed a comparison on the \texttt{ACC\_Large} \ac{NN} as well as the retrained \texttt{ACC\_Large} \ac{NN}, i.e. on a satisfiable as well as a non-satisfiable instance.
We compared our approach to dReal~\cite{dReal}, Mathematica~\cite{Mathematica}, Z3~\cite{Z3overall}, MathSAT~\cite{mathsat} (due to its use of incremental linearization) as well as the first and second place of SMT-Comp 2023 in the \texttt{QF\_NRA} track: Z3++~\cite{z3pp1} and  cvc5~\cite{cvc5}.
The results of our comparison can be observed in \Cref{tab:smt_comp}.
The observed timeouts after 12 hours are unsurprising insofar as the work on linear \ac{OLNNV} techniques was partially motivated by the observation that classical SMT solvers struggle with the verification of \acp{NN}.
\begin{table}
    \centering
    \begin{tabular}{c|c|c|c|c}
        \multirow{2}{*}{Tool} & \multicolumn{2}{c|}{\texttt{ACC\_Large}} & \multicolumn{2}{c}{\texttt{ACC\_Large} retrained}\\\cline{2-5}
        & Status & Time & Status & Time\\\hline\hline
        Mathematica & MO & --- & MO & ---\\
        dReal & TO & --- & TO & ---\\
        Z3 & unknown & \hphantom{2}510s & unknown & 1793s\\
        Z3++ & unknown & 2550s & unknown & 2269s\\
        cvc5 & TO & --- & TO & ---\\
        MathSAT & TO & --- & TO & ---\\
        \ac{SNNT} & \textbf{sat} & \hphantom{25}\textbf{87s} & \textbf{unsat} & \hphantom{1}\textbf{124s}
    \end{tabular}
    \caption{Comparison of \ac{SNNT} with State-of-the-Art SMT solvers: Timeout (TO) was set to 12 hours}
    \label{tab:smt_comp}
\end{table}

\paragraph{Comparison to the techniques by Genin et al.~\cite{Genin2022}}
While the work by Genin et al.~\cite{Genin2022} represents a case-study with techniques specifically applied to an \ac{NN} for a simplified airborne collision avoidance setting,
some ideas from the example in \cite{Genin2022} might in principle generalize to other case studies.
Unfortunately, the case-study considered by \cite{Genin2022} are not the \acp{NN} from Julian et al.~\cite{julian2016policy,julian2020safe}, but simplified \acp{NN} with a single acceleration control output.
As the authors did not publish their trained NNs, their exact verification formulas, or their verification runtimes, we instead compare our approach with this line of work on our ACC benchmark.
To this end, we approximate the verification property derived in \Cref{subsec:eval:acc} using the box approximation techniques described by the authors and use their Lantern Python package to translate the verification tasks into linear arithmetic SMT problems.
Using Z3, their technique does not terminate within more than 50 hours on the (unsafe) ACC\_Large network and thus fails to analyze the NN. This demonstrates significant scalability limitations compared to our approach.
Moreover, it is worth pointing out that the authors themselves acknowledge that the technique is incomplete which distinguishes our complete lifting procedure from their approach.

\paragraph{Results on Q2.}
\looseness=-1
If \ac{CLNNV} is a hammer then guaranteeing infinite-time safety is a screw:
It is a categorically different problem requiring a different tool.
\ac{SNNT} provides safety guarantees which go infinitely beyond the guarantees achievable with State-of-the-Art techniques (\ac{CLNNV} or otherwise).
A direct CAD/SMT encoding of \Cref{eq:safety_property}
as well as the techniques by Genin et al.~\cite{Genin2022} are no alternatives
due to timeouts ($>$12h) or ``unknown'' results (see also \Cref{tab:smt_comp}).

\subsection{Zeppelin Steering}
\label{subsec:eval:zeppelin}
As a further case study, we considered the task of steering a Zeppelin under uncertainty:
The model's goal was to learn avoiding obstacles while flying in a wind field with nondeterministic wind turbulences.
This problem has previously been studied with differential hybrid games~\cite{Platzer2017}.
The examined scenario serves two purposes:
On the one hand, it shows that our approach can reuse safety results from the \ac{dL} literature which drastically increases its applicability;
on the other hand it is a good illustration for why verification (rather than empirical evidence) is so important when deploying \acp{NN} in safety-critical fields.

After transferring the differential hybrid games logic contract into a differential dynamic logic contract and proving its safety, we trained a model to avoid obstacles while flying in a wind field with uniformly random turbulences via PPO.
After 1.4 million training steps, we obtained an agent that did not crash for an evaluation run of 30,000 time steps.
Given these promising results we proceeded to verify the agent's policy assuming a safe -- or at least ``almost safe'' -- flight strategy had been learnt.
However, upon verifying the \ac{NN}'s behavior for obstacles of circumfence 40, we found that it produced potentially unsafe actions for large parts of the input space.
The reason this unsafety was not observable during empirical evaluation was the choice of uniformly random wind turbulences:
The unsafe behavior only appears for specific sequences of turbulences which occur extremely rarely in the empirical setting.
This flaw in the training methodology was only found due to the verification.
This is where our approach differs from simulation-based evaluation:
With an SMT filter timeout of 4 seconds, \ac{SNNT} provides an \emph{exhaustive} characterization of \emph{all} potentially unsafe regions in 4.1 hours while providing 72 concrete counterexample regions.
This is where our approach differs from simulation-based evaluation, as we were able to generate an exhaustive characterization of counterexample regions.
In this instance, the tool's bootlenecks were approximation construction and the SMT based counterexample finding.
This case study and the stark difference between simulation and verification underscore the importance of rigorous verification of \acp{NN} as an addition to empirical evidence in safety-critical areas.

\section{Allowed advisories for Vertical Airborne Collision Avoidance}
\label{apx:acas_table}
\Cref{tab:acas_table} provides an overview of possible advisories for a Vertical Airborne Collision Avoidance System.
The allowed range of vertical velocity and the required minimal acceleration are integrated into the \ac{dL} model used by \chThreeFancyName{}.
\begin{table}[!h]
    \centering
    \begin{tabularx}{\textwidth}{l | X | l | l}
        Advisory & Description & Vertical Velocity & Min. Acceleration\\
        && [ft/min] & \\\hline
        COC & Clear of conflict & --- & ---\\
        DNC & Do not climb & $\left[-\infty,0\right]$ & $g/4$\\
        DND & Do not descend & $\left[0,\infty\right]$ & $g/4$\\
        DES1500 & Descend at least 1500 ft/min & $\left[-\infty,-1500\right]$ & $g/4$\\
        CL1500 & Climb at least 1500 ft/min & $\left[1500,\infty\right]$ & $g/4$\\
        SDES1500 & Strengthen descent to at least 1500 ft/min & $\left[-\infty,-1500\right]$ & $g/3$\\
        SCL1500 & Strengthen climb to at least 1500 ft/min & $\left[1500,\infty\right]$ & $g/3$\\
        SDES2500 & Strengthen descent to at least 2500 ft/min & $\left[-\infty,-2500\right]$ & $g/3$\\
        SCL2500 & Strengthen climb to at least 2500 ft/min & $\left[2500,\infty\right]$ & $g/3$
    \end{tabularx}
    \caption{Overview on Vertical Airborne Collision Advisories (simplified version of \cite[Table 1]{Jeannin2017})}
    \label{tab:acas_table}
\end{table}

\section{\acp{NMAC} produced by \ac{NN}-based \ac{ACASX} advisories}
\label{apx:acas}
Further counterexamples for the advisories of the \ac{NNCS} can be found in \Cref{fig:acas:des1500,fig:acas:sdes1500,fig:acas:scl1500,fig:acas:sdes2500,fig:acas:scl2500}.
\begin{figure}[!h]
    \centering
    \includegraphics[width=\textwidth]{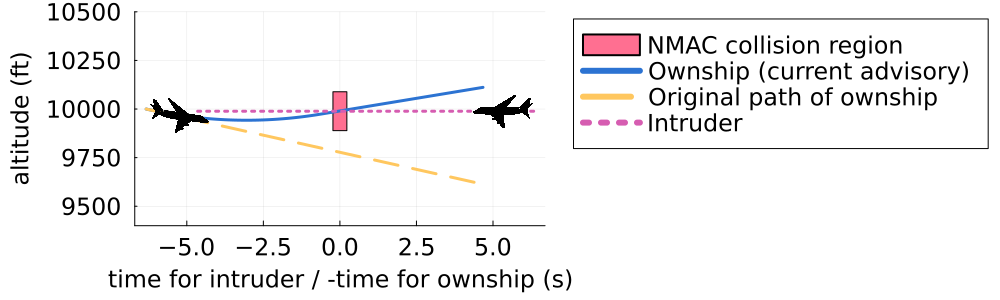}
    \caption{After a previous advisory to descend at least 1500ft/min, the NN advises the pilot to strengthen climb to at least 1500ft/min leading to a NMAC.}
    \label{fig:acas:des1500}
\end{figure}
\begin{figure}[!h]
    \centering
    \includegraphics[width=\textwidth]{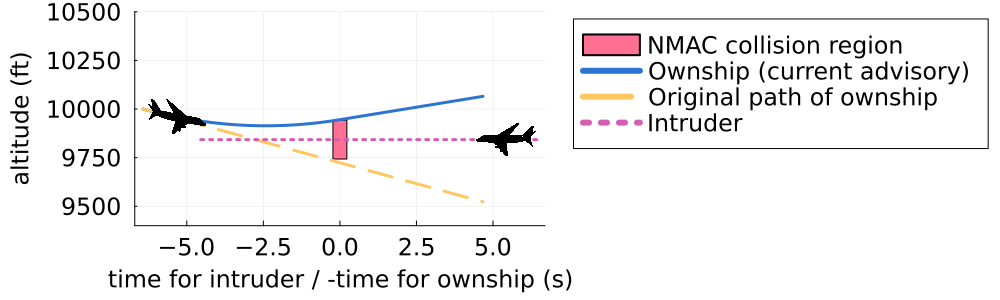}
    \caption{After a previous advisory to strengthen descent to at least 1500ft/min, the NN advises the pilot to strengthen climb to at least 1500ft/min leading to a NMAC.}
    \label{fig:acas:sdes1500}
\end{figure}
\begin{figure}[!h]
    \centering
    \includegraphics[width=\textwidth]{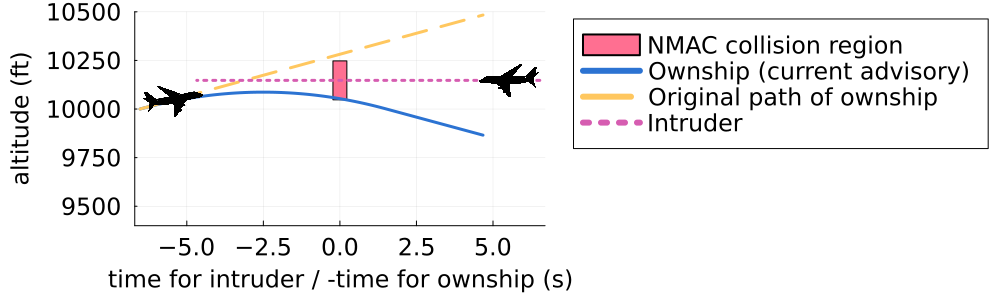}
    \caption{After a previous advisory to strengthen climb to at least 1500ft/min, the NN advises the pilot to strengthen descent to at least 2500ft/min leading to a NMAC.}
    \label{fig:acas:scl1500}
\end{figure}
\begin{figure}[!h]
    \centering
    \includegraphics[width=\textwidth]{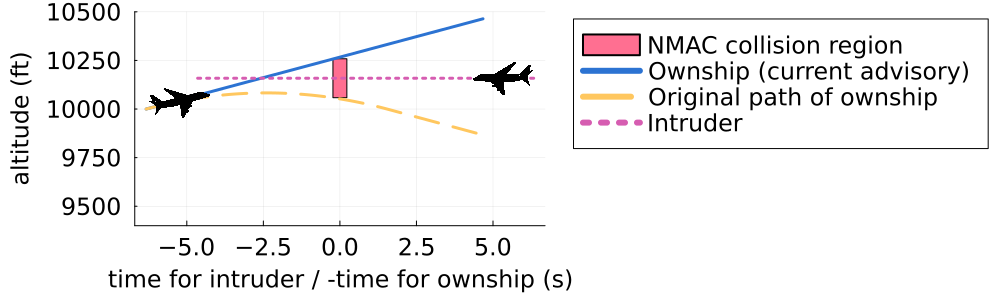}
    \caption{After a previous advisory to strengthen descent to at least 2500ft/min, the NN advises the pilot to strengthen descent to at least 2500ft/min leading to a NMAC.}
    \label{fig:acas:sdes2500}
\end{figure}
\begin{figure}[!h]
    \centering
    \includegraphics[width=\textwidth]{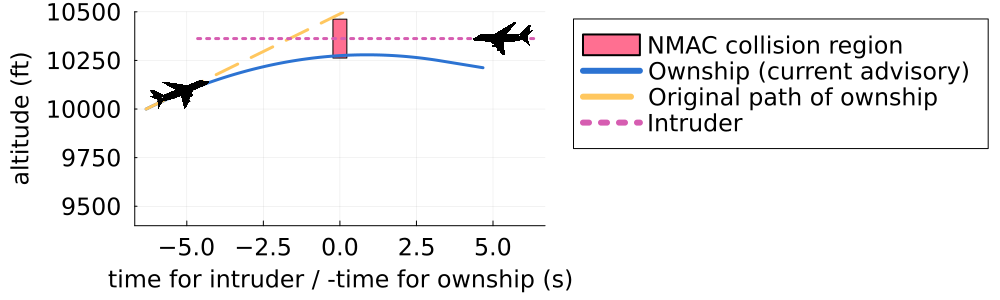}
    \caption{After a previous advisory to strengthen climb to at least 2500ft/min, the NN advises the pilot to strengthen descent to at least 1500ft/min leading to a NMAC.}
    \label{fig:acas:scl2500}
\end{figure}

Counterexamples to the safety of \ac{NNCS} advisories for non-level flight of the intruder in the case of a previous advisory \emph{Do Not Climb} and \emph{Do Not Descend} can be found in \Cref{fig:acas:vI:dnc,fig:acas:vI:dnd}.
Note, that for non-level flight (i.e. both intruder and ownship have a non-zero vertical velocity), there exist two possible interpretations for the advised vertical velocities.
These can be interpreted as absolute or relative velocity.
For our counterexamples in \Cref{fig:acas:vI:dnc,fig:acas:vI:dnd} we opt for the relative velocity interpretation.
This does not affect the analysis for level flight intruders where the interpretations coincide.

\begin{figure}[!h]
    \centering
    \includegraphics[width=\textwidth]{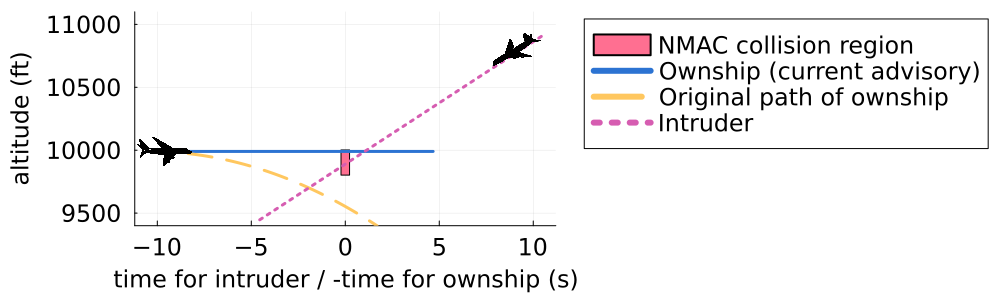}
    \caption{After a previous advisory to not climb, the NN advises the pilot to climb with at least 1500ft/min leading to a NMAC (assumes relative velocity interpretation).}
    \label{fig:acas:vI:dnc}
\end{figure}
\begin{figure}[b]
    \centering
    \includegraphics[width=\textwidth]{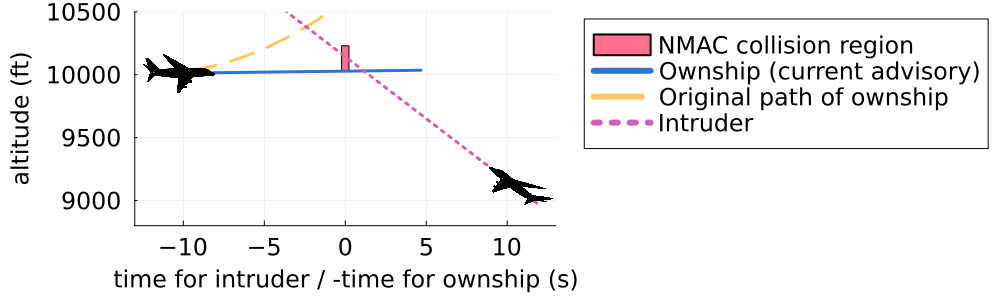}
    \caption{After a previous advisory to not descend, the NN advises the pilot to descend with at least 1500ft/min leading to a NMAC (assumes relative velocity interpretation).}
    \label{fig:acas:vI:dnd}
\end{figure}

\FloatBarrier

\clearpage

\ifchecklist
\section*{NeurIPS Paper Checklist}

\begin{enumerate}

\item {\bf Claims}
    \item[] Question: Do the main claims made in the abstract and introduction accurately reflect the paper's contributions and scope?
    \item[] Answer: \answerYes{} %
    \item[] Justification: The claims made in the introduction are supported by the theorems (with proofs), algorithms and evaluations in \Cref{sec:safe_controller_implementations,sec:nnnv,sec:evaluation} (and its appendices).%
\item {\bf Limitations}
    \item[] Question: Does the paper discuss the limitations of the work performed by the authors?
    \item[] Answer: \answerYes{} %
    \item[] Justification: The assumptions made by our theoretical results (in particular in \Cref{sec:safe_controller_implementations}) are explicitly stated and the exponential worst-case runtime of \chFourFancyName{} is discussed in \Cref{sec:nnnv}.

\item {\bf Theory Assumptions and Proofs}
    \item[] Question: For each theoretical result, does the paper provide the full set of assumptions and a complete (and correct) proof?
    \item[] Answer: \answerYes{} %
    \item[] Justification: All formal statements and proofs can be found in \Cref{apx:proofs} and are referenced accordingly in the main paper.

    \item {\bf Experimental Result Reproducibility} %
    \item[] Question: Does the paper fully disclose all the information needed to reproduce the main experimental results of the paper to the extent that it affects the main claims and/or conclusions of the paper (regardless of whether the code and data are provided or not)?
    \item[] Answer: \answerYes{} %
    \item[] Justification: We provide an artifact for the \ac{SNNT} tool which comprises all benchmarks mentioned in the paper.
    The artifact is equipped with scripts for building the tool (\texttt{./build.sh <julia 1.10 path>}), running an example (\texttt{./run\_example.sh}; runs approx. 130 seconds) or running all \ac{SNNT} benchmarks (\texttt{./run\_experiments.sh}; logs will be saved in \texttt{./experiments/<benchmark class>/<benchmark>}).
    For traceability, the benchmark logs are part of the artifact (will be overwritten by rerunning benchmarks).
    For the comparisons to other tools in \Cref{subsec:eval:closed_loop}, we provide logs and documentation of the evaluated configurations (see \texttt{./experiments/comparison} and \texttt{./experiments/dnnv\_comp}).
    For all benchmarks, the underlying \ac{dL} models are provided (\texttt{./experiments/<benchmark class>/<benchmark>/keymaerax}).
    Where \acp{NN} were trained, we also provide the Jupyter notebooks used for training as-is (\texttt{./experiments/<benchmark class>/<benchmark>/training}).

\item {\bf Open access to data and code}
    \item[] Question: Does the paper provide open access to the data and code, with sufficient instructions to faithfully reproduce the main experimental results, as described in supplemental material?
    \item[] Answer: \answerYes{} %
    \item[] Justification: All code and data can be found on GitHub: \url{https://github.com/samysweb/NCubeV}. We also provide an archived version of our artifact with DOI~\cite{samuel_teuber_2024_13922169}.

\item {\bf Experimental Setting/Details}
    \item[] Question: Does the paper specify all the training and test details (e.g., data splits, hyperparameters, how they were chosen, type of optimizer, etc.) necessary to understand the results?
    \item[] Answer: \answerYes{} %
    \item[] Justification: We provide details on the machine used for evaluation; all other details are described in the paper and found in the artifact.

\item {\bf Experiment Statistical Significance}
    \item[] Question: Does the paper report error bars suitably and correctly defined or other appropriate information about the statistical significance of the experiments?
    \item[] Answer: \answerNo{} %
    \item[] Justification: We did not have the computational resources to run the experiments multiple times (this is in particular the case for long-running experiments).
    More importantly, our results show fundamental, categorical differences between the tools:
    We repeatedly demonstrate that our tool addresses problems that all other techniques fail to address within a very generous timeout.
    Thus, a statistical significance test would not be useful in this matter as the results are binary (solves the problem vs. does not solve the problem).

\item {\bf Experiments Compute Resources}
    \item[] Question: For each experiment, does the paper provide sufficient information on the computer resources (type of compute workers, memory, time of execution) needed to reproduce the experiments?
    \item[] Answer: \answerYes{} %
    \item[] Justification: All experiments ran on the same machine (see \Cref{sec:evaluation})
    
\item {\bf Code Of Ethics}
    \item[] Question: Does the research conducted in the paper conform, in every respect, with the NeurIPS Code of Ethics \url{https://neurips.cc/public/EthicsGuidelines}?
    \item[] Answer: \answerYes{} %
    \item[] Justification: The core contribution of our research is a methodology to make \ac{NNCS} safer. Making \ac{NNCS} safer can have very positive societal impacts by increasing (warranted) trust in infrastructure and preventing catastrophic failures. We do not see potential issues in any other category mentioned by the NeurIPS Code of Ethics.

\item {\bf Broader Impacts}
    \item[] Question: Does the paper discuss both potential positive societal impacts and negative societal impacts of the work performed?
    \item[] Answer: \answerNo{} %
    \item[] Justification: We see our work as foundational research that could be applied in numerous settings to guarantee safety of \ac{NNCS}. We see no potential for negative societal impacts directly stemming from this work.
    Conversely, making \ac{NNCS} safer can have very positive societal impacts by increasing (warranted) trust in infrastructure and preventing catastrophic failures.
    
\item {\bf Safeguards}
    \item[] Question: Does the paper describe safeguards that have been put in place for responsible release of data or models that have a high risk for misuse (e.g., pretrained language models, image generators, or scraped datasets)?
    \item[] Answer: \answerNA{} %
    \item[] Justification: Our contribution does not contain any high-risk datasets or models.

\item {\bf Licenses for existing assets}
    \item[] Question: Are the creators or original owners of assets (e.g., code, data, models), used in the paper, properly credited and are the license and terms of use explicitly mentioned and properly respected?
    \item[] Answer: \answerYes{} %
    \item[] Justification: Evaluated models from the literature are appropriately cited.

\item {\bf New Assets}
    \item[] Question: Are new assets introduced in the paper well documented and is the documentation provided alongside the assets?
    \item[] Answer: \answerYes{} %
    \item[] Justification: The code of our tool \ac{SNNT} will be released as open-source with the paper -- this includes instructions for verifying models with the tool.

\item {\bf Crowdsourcing and Research with Human Subjects}
    \item[] Question: For crowdsourcing experiments and research with human subjects, does the paper include the full text of instructions given to participants and screenshots, if applicable, as well as details about compensation (if any)? 
    \item[] Answer: \answerNA{} %
    \item[] Justification: No human subjects.

\item {\bf Institutional Review Board (IRB) Approvals or Equivalent for Research with Human Subjects}
    \item[] Question: Does the paper describe potential risks incurred by study participants, whether such risks were disclosed to the subjects, and whether Institutional Review Board (IRB) approvals (or an equivalent approval/review based on the requirements of your country or institution) were obtained?
    \item[] Answer: \answerNA{} %
    \item[] Justification: No human subjects.

\end{enumerate}
\fi

\end{document}